\documentclass[12pt]{article}
\usepackage[onehalfspacing]{setspace}
\usepackage{array,xspace,multirow,hhline,tikz,colortbl,tabularx,booktabs,fixltx2e,amsmath,amssymb,amsfonts,amsthm}
\usetikzlibrary{arrows}
\usetikzlibrary{shapes}
\usepackage{algorithm,algorithmic}
\usepackage{eqparbox}
\usepackage{verbatim,ifthen}
\usepackage{enumitem}
\usepackage{pifont} 
\usepackage{setspace}
\usepackage{graphicx}
\usepackage{caption} 
\usepackage{pgfplots}
\usepackage{bbm}
\usepackage{wrapfig} 
\usepackage{natbib} 
\usepackage{ifthen}
\usepackage{calrsfs,mathrsfs}
\usepackage{bbding,pifont} 
\usepackage{pgflibraryshapes} 
\usepackage{url}  
\usepackage{hyperref} 
\usepackage{geometry}
\usepackage{abstract}
 \usepackage[bottom,flushmargin,hang,multiple]{footmisc}

\usepackage{gensymb} 

\usepackage{array,xspace,hhline,tikz,colortbl,tabularx,fixltx2e,amsmath,amssymb,amsfonts,amsthm}
\usepackage{verbatim,ifthen}
\usepackage{pifont}
\usepackage{ifthen}
\usepackage{pgflibraryshapes}
\usepackage[multiple]{footmisc}
\usepackage{nicefrac}
\usetikzlibrary{arrows}
\usepackage{ltxtable}
\usepackage{longtable}

	\usepackage{varioref}

\tikzset{
  jumpdot/.style={mark=*,solid},
  excl/.append style={jumpdot,fill=white},
  incl/.append style={jumpdot,fill=black},
  rexcl/.append style={jumpdot,color=red,fill=white},
  rincl/.append style={jumpdot,fill=black,color=red},
}

	\usepackage{varioref} 
		\usepackage{subfigure}

\definecolor{light-gray}{gray}{0.9}

\newtheorem{definition}{Definition}
\newtheorem{innercustomgeneric}{\customgenericname}
\providecommand{\customgenericname}{}
\newcommand{\newcustomtheorem}[2]{%
  \newenvironment{#1}[1]
  {%
   \renewcommand\customgenericname{#2}%
   \renewcommand\theinnercustomgeneric{##1}%
   \innercustomgeneric
  }
  {\endinnercustomgeneric}
}

\newcustomtheorem{customthm}{Theorem}
\newcustomtheorem{customlemma}{Lemma}
\newcustomtheorem{customproposition}{Proposition}

\DeclareMathOperator*{\argmax}{arg\,max}

	\newtheorem{lemma}{Lemma}%
	\newtheorem{theorem}{Theorem}%
	\newtheorem{proposition}{Proposition}%
	\newtheorem{corollary}{Corollary}%
\hypersetup{
	colorlinks=true,
	linkcolor=blue,
	filecolor=blue,      
	urlcolor=blue,
	citecolor=blue
}

\usepackage[normalem]{ulem}
	\newcommand\eat[1]{}

	\usepackage{enumitem}
	\setenumerate[1]{label=\rm(\it{\roman{*}}\rm),ref=({\it\roman{*}}),leftmargin=*}
	\newlength{\wordlength}

	\newcommand{\eqclass}[2][]{\ifthenelse{\equal{#1}{}}{[#2]}{[#2]_{\sim_{#1}}}}

	\newcommand{\ceil}[1]{\lceil #1 \rceil }
	\newcommand{\floor}[1]{\lfloor #1 \rfloor }

\usepackage{enumitem}
\setenumerate[1]{label=\rm(\it{\roman{*}}\rm),ref=({\it\roman{*}}),leftmargin=*}

\newcommand{\nbh}[1][]{
	\ifthenelse{\equal{#1}{}}{\nu}{\nu(#1)}
}

\newcommand{\cstr}[1][]{
	\ifthenelse{\equal{#1}{}}{\mathscr S}{\cstr(#1)}
}

\newcommand{\choice}[1][]{
	\ifthenelse{\equal{#1}{}}{\mathit{C}}{\choice(#1)}

		\newcommand{\ml}[1][]{\ensuremath{\ifthenelse{\equal{#1}{}}{\mathit{ML}}{\mathit{ML}(#1)}}\xspace}
		\newcommand{\sml}[1][]{\ensuremath{\ifthenelse{\equal{#1}{}}{\mathit{SML}}{\mathit{SML}(#1)}}\xspace}
		\newcommand{\sd}[1][]{\ensuremath{\ifthenelse{\equal{#1}{}}{\mathit{SD}}{\mathit{SD}(#1)}}\xspace}
		\newcommand{\rsd}[1][]{\ensuremath{\ifthenelse{\equal{#1}{}}{\mathit{RSD}}{\mathit{RSD}(#1)}}\xspace}
		\newcommand{\rd}[1][]{\ensuremath{\ifthenelse{\equal{#1}{}}{\mathit{RD}}{\mathit{RD}(#1)}}\xspace}
		\newcommand{\st}[1][]{\ensuremath{\ifthenelse{\equal{#1}{}}{\mathit{ST}}{\mathit{ST}(#1)}}\xspace}
		\newcommand{\bd}[1][]{\ensuremath{\ifthenelse{\equal{#1}{}}{\mathit{BD}}{\mathit{BD}(#1)}}\xspace}
		\newcommand{\pc}[1][]{\ensuremath{\ifthenelse{\equal{#1}{}}{\mathit{PC}}{\mathit{PC}(#1)}}\xspace}
		\newcommand{\dl}[1][]{\ensuremath{\ifthenelse{\equal{#1}{}}{\mathit{DL}}{\mathit{DL}(#1)}}\xspace}
		\newcommand{\ul}[1][]{\ensuremath{\ifthenelse{\equal{#1}{}}{\mathit{UL}}{\mathit{UL}(#1)}}\xspace}
		
			\newcommand{\indiff}{\ensuremath{\sim}}}

			\usepackage{palatino}

\begin{document}

	\title{Strategyproof and Proportionally Fair \\ Facility Location}

	\author{Haris Aziz\thanks{UNSW Sydney, Australia. Email: \href{mailto: haris.aziz@unsw.edu.au}{haris.aziz@unsw.edu.au}}  \ \ Alexander Lam\thanks{City University of Hong Kong, Hong Kong. Email: \href{mailto: alexlam@cityu.edu.hk}{alexlam@cityu.edu.hk}} \ \ Barton E. Lee\thanks{ETH Z\"urich, Switzerland. Email: \href{mailto: barton.e.lee@gmail.com}{barton.e.lee@gmail.com}}  \ \ Toby Walsh\thanks{UNSW Sydney and Data61 CSIRO, Australia. Email: \href{mailto: t.walsh@unsw.edu.au}{t.walsh@unsw.edu.au}}\\ \vspace{2mm} \\\small{\href{http://www.cse.unsw.edu.au/~haziz/fairFLP.pdf}{Latest version here}} \vspace{2mm} \\  \small{ \textbf{First posted}: 22nd October 2021.} }

	\maketitle

		\begin{abstract} 
		We focus on a simple, one-dimensional collective decision problem (often referred to as the facility location problem) and explore issues of strategyproofness and proportionality-based fairness. We introduce and  analyze a hierarchy of proportionality-based fairness axioms   of varying strength: Individual Fair Share (IFS), Unanimous Fair Share (UFS), Proportionality \citep[as in][] {FPPV21}, and Proportional Fairness (PF). For each axiom, we characterize the family of  mechanisms that satisfy the axiom and strategyproofness. We show that imposing strategyproofness renders many of the axioms to be equivalent: the family of mechanisms that satisfy proportionality, unanimity, and strategyproofness is equivalent to  the family of mechanisms that satisfy UFS and strategyproofness, which, in turn, is equivalent to the family of mechanisms  that satisfy     PF and strategyproofness. Furthermore, there is a unique such mechanism: the Uniform Phantom mechanism, which is studied in~\cite{FPPV21}. We also characterize the outcomes of the Uniform Phantom mechanism  as  the  unique  (pure) equilibrium outcome for any mechanism that satisfies continuity, strict monotonicity, and UFS. Finally, we analyze the approximation guarantees, in terms of optimal social welfare and minimum total cost, obtained by mechanisms that are strategyproof and satisfy each proportionality-based fairness axiom. We show that the Uniform Phantom mechanism provides the best  approximation of the optimal social welfare (and also minimum total cost) among all mechanisms that satisfy UFS.
			\end{abstract}
	
%
%
%
%
%
	
		\clearpage
\section{Introduction.} 
		
 
Facility location problems are ubiquitous in society and capture various collective scenarios. Examples include electing political representatives~\citep*{BoJo83,FFG16,Moul80}, selecting policies~\citep*{BaAn21,DrLa19,KMN17a}, deciding how to allocate a public budget~\citep*{FPPV21}, and deciding the location or services provided by public facilities~\citep{ScVo02}. Two key concerns in such problems are that the selection process may be vulnerable to strategic manipulations and/or fail to guarantee ``fair'' outcomes. In this paper, we simultaneously examine the issues of strategyproofness and fairness for the facility location problem.  		
 		
		In the facility location problem, each agent is viewed as a point on an interval. Depending on the motivating setting, the point could reflect the agent's physical location, political position, or social preference.  Each agent has symmetrically single-peaked preferences and prefers the collective outcome to be near their own position. The goal of the collective decision problem is to take agents' preferences (positions) into account to find a reasonable collective outcome (the location of the facility). 

		The facility location problem (or the one-dimensional collective decision problem) is one of the most fundamental problems in economics, computer science, and operations research. It takes a central place in social choice theory as single-peaked preferences are one of the key preference restrictions that circumvent the infamous Gibbard-Satterthwaite theorem~\citep{Gibb73,Satt75}, which says that in general social choice, no unanimous and non-dictatorial voting mechanism is strategyproof. Furthermore, despite the unidimensional setting  appearing restrictive, it is well suited to many real-world problems---most prominently, deciding the level of provision of a public good~\citep[see, e.g.,][]{BaJa94,Cant04}. When agents have single-peaked preferences, the mechanism that returns the median voter's position is unanimous, non-dictatorial, and strategyproof~\citep[see,][]{Moul80}. 		This seminal result has been discussed in hundreds of papers. 
		Despite the importance of the median mechanism for the facility location problem, it does not satisfy several fairness concepts that are inspired from the theory of fair division and proportional representation. 
		 We focus on the following research questions. 
		 \begin{quote} 
			\emph{ For the facility location problem, what are natural fairness concepts? How well can these fairness concepts be achieved by strategyproof mechanisms? For strategyproof mechanisms that satisfy one of these fairness concepts, which mechanism performs optimally in terms of maximizing social welfare  or minimizing total cost? Which mechanisms achieve fairness in equilibrium?}
			 \end{quote}

		 Our contributions are four-fold. First, we consolidate a number of fairness axioms from the literature, explicitly describe their relations and establish the compatibility---and, in some cases, incompatibility---of strategyproofness with these fairness concepts.  We propose a new concept called \emph{proportional fairness (PF)} that is based on the idea that the distance of a facility from a group of agents should depend both on the size of the group as well as how closely  the agents are clustered. We also analyze existing  axioms from the literature on fair division, participatory budgeting, and proportional representation such as    \emph{proportionality}, \emph{unanimous fair share (UFS)}, \emph{individual fair share (IFS)}, and unanimity.  Our PF axiom  is the strongest of these; Figure~\ref{Figure: relation} describes the relationship between all the fairness axioms that we study.

%
%
%
%
%
%
%
%
%

\begin{figure}[!htb]
    \centering
 
        \includegraphics[ height=0.2\textheight]{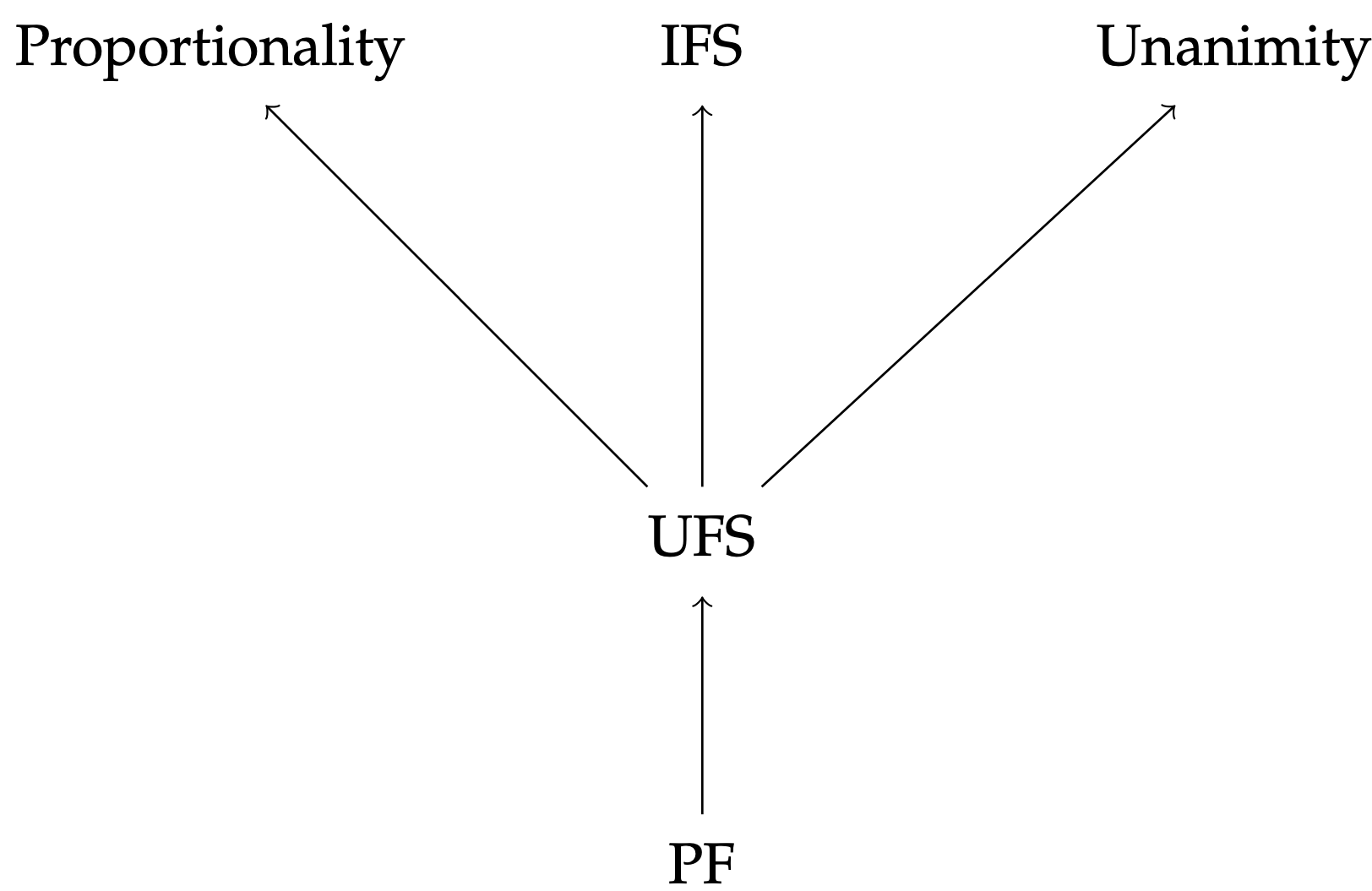}
        \caption{{\footnotesize{Relations between axioms. 
			 An arrow from (A)   to (B) denotes that  (A) implies (B). All relations are strict.}}}
		 \label{Figure: relation}
 \end{figure}

Second, we present two characterization results. We characterize the family of strategyproof mechanisms that satisfy unanimity, anonymity, and IFS.
We then identify a specific mechanism, called the Uniform Phantom mechanism, that uniquely satisfies strategyproofness, unanimity, and proportionality. We also prove that the Uniform Phantom mechanism uniquely satisfies strategyproofness and UFS. Since we show that the Uniform  Phantom mechanism also satisfies PF (and because PF implies UFS), we obtain as a corollary  that the Uniform Phantom mechanism is the only strategyproof mechanism satisfying PF. Therefore, within the class of strategyproof mechanisms, PF and UFS collapse to the same property---in contrast,  IFS is markedly weaker, even within the class of strategyproof mechanisms.

		Third, we consider the fairness of outcomes under strategic behavior when a mechanism is not strategyproof. We prove that if a mechanism satisfies continuity, strict monotonicity, and UFS, then a pure Nash equilibrium exists, and every (pure) equilibrium under the mechanism satisfies UFS with respect to agents' true locations. One mechanism in this class is the Average mechanism, which locates the facility at the average of all agents' reported locations. Furthermore,  for   mechanisms satisfying continuity, strict monotonicity, and UFS, the equilibrium outcome leads to a facility location that equals the facility location of the Uniform Phantom mechanism when agents report their true location. Thus, our equilibrium analysis of continuous, strictly monotonic, and UFS mechanisms provides an alternative characterization of the Uniform Phantom mechanism.

Lastly, we take an approximate mechanism design perspective~\citep{NiRo01a,PrTe13}.  We explore how well the maximum social welfare and minimum total cost can be approximated when fairness axioms and strategyproofness are imposed. Our goal is to identify mechanisms that deliver the best approximation guarantees while also satisfying strategyproofness and the corresponding fairness axioms (such as IFS and UFS). We first establish a stark negative result for the total cost approximation. Any strategyproof, anonymous, and unanimous mechanism that satisfies IFS has approximation ratio of $n-1$, which is unbounded as $n$ grows. Because IFS is our weakest fairness axiom when strategyproofness is imposed, the total cost approximation analysis fails to distinguish any difference between the mechanisms that we focus on. We then turn to social welfare approximation where we establish more positive and nuanced results. Intuitively, imposing UFS leads to a strictly worse approximation ratio than if only IFS   is imposed and, in either case, the best approximation guarantee is bounded. We identify   strategyproof mechanisms that provide the best approximation of the maximum social welfare among all (not necessarily strategyproof) mechanisms that satisfy either IFS or UFS. In the latter case  of satisfying UFS, the Uniform Phantom mechanism achieves this best approximation. In this sense, the fairness axioms impose a greater cost  on the social welfare approximation guarantees than the strategyproofness requirement.   
 


\subsection{Related literature.}\label{section: related lit} 

 
\paragraph{Facility location problems.} 
The facility location problem has been studied extensively in operations research, economics, and computer science. As is common in the economics literature, our paper takes a mechanism design approach. We assume an incomplete information setting, where agents have privately-known utility functions (and, hence, peak locations) and can strategically (mis)report their peak location. The problem is to design a mechanism that is strategyproof and achieves a ``desirable'' facility location with respect to the agents' true locations.  \citeauthor{Moul80}'s (\citeyear{Moul80}) seminal work characterizes the family of strategyproof and Pareto efficient mechanisms when agents have single-peaked preferences. 
In our paper, agents have single-peaked preferences that are also \emph{symmetric}, i.e., agents prefer the facility to be located closer to their location regardless of whether it is to left or right of their location; therefore, our setting is closer to  \citet{BoJo83}. \citeauthor{BoJo83} characterize a strict subfamily of strategyproof mechanisms, which includes the family of strategyproof and unanimous mechanisms (\citeauthor{BoJo83} also extend their results to higher dimensions).  \citet{MaMo11a} formalize the connection between the mechanism design problem in settings where agents have single-peaked preferences and settings where agents have symmetrically single-peaked preferences. 

Since   \citet{Moul80} and \citet{BoJo83}, numerous scholars have explored  open-questions related to these characterizations~\citep*[see, e.g.,][]{BaJa94,BMS98,Chin97,JLPV21,MaMo11a,PPS+97,Weym11}. Others have explored extensions and variations of the facility location problem. For example, \citet{NePu06,NePu07a} relax the assumption that agents have single-peaked preferences;~\citet{Miya98,Miya01} and~\cite{Ehle02,Ehle03} extend the facility location problem to consider locating multiple facilities; ~\citet*{ACLL+20,ACLP20} introduce capacity constraints into the problem;~\citet{JaNi04} introduce interdependent utilities; \citet{Cant04} introduces an outside option; and~\citet{ScVo02} extend the facility location problem to a network setting. For a recent survey of the computational social choice literature on facility location problems, see~\citet*{CF-RL+21}. Our paper contributes to this literature by formalizing a hierarchy of ``proportionality-based fairness'' axioms for the facility location problem and characterizing families of strategyproof and fair mechanisms within each layer of the heirarchy. Additionally, in Section~\ref{section: equilibrium properties}, we explore the equilibrium properties of non-strategyproof mechanisms. We obtain results that  complement those of~\citet{ReTr05a,ReTr11a} and~\citet{YaKa13a} (further details   provided in Section~\ref{section: equilibrium properties}).

There is also an extensive literature in operations research and computer science that studies the facility location problem within a complete information setting. These literatures largely focus on issues of computational complexity and approximation and, therefore, are not directly relevant to the present paper~\citep[for an overview, see][]{BrCh89,FaHe09a}.

\paragraph{Fairness in collective decision problems.} Issues of fairness in collective decision problems have been studied in a variety of contexts~\citep[see, e.g.,][]{Dumm97a,Mill61,Nash50,Nash53,Rawls71a,Sen1980,Shap53,Yaar81}. Most closely related to the present paper are the social choice and computational social choice literatures~\citep*[for an overview, see][]{ASS10,ABES19,Endr17,FSST17a,Klam10,LaSa10}. We formalize a hierarchy of fairness axioms for the facility location problem that are conceptually related to proportional representation. As will be discussed in Section~\ref{section: defense}, our axioms can also be motivated by---and connect with---notions of stability in cooperative game theory, such as the ``core''~\citep[see, e.g., ][]{Scar67}.  Two of our fairness axioms (IFS and UFS) are translations of the ``individual fair share'' and ``unanimous fair share'' axioms, which appear in fair division and participatory budgeting problems~\citep*{ABM19,Moul03}, into the facility location problem. In addition, we utilize a natural axiom of proportional representation, called ``proportionality'', which is explored in the context of participatory budgeting by~\citet*{FPPV21}.  Beyond translating existing notions of fairness into the facility location problem, we also introduce the new axiom of ``Proportional Fairness'' that is stronger than all of the aforementioned axioms. 

Our approach contrasts with a number of facility location papers that attempt to obtain outcomes that achieve (or approximate) the egalitarian outcome, i.e., maximizing the utility of the worst off agent~\citep[see, e.g.,][]{PrTe13}. 
 {\citet{Mull91a} notes that the egalitarian objective is sensitive to extreme locations and recommends distributional equality as an underlying principle for considering equality measures. When placing multiple facilities, several new concepts have been proposed for capturing proportionality-based fairness  concerns~\citep[see, e.g.,][]{BiDe00a,JKL20a}. However, these concepts are equivalent to weak Pareto optimality or unanimity when there is only one facility.} For the single-facility problem, \citet*{ZLC22} recently examined the issue of welfare guarantees for groups of agents. Our approach and results differ in that we consider the classic facility location problem whereas \citeauthor{ZLC22} overlay it with additional information that places agents in predetermined groups.

In the context of the facility location problem, our paper characterizes strategyproof and ``fair'' mechanisms. Some of our results directly relate to those of~\citet{FPPV21}. In the context of participatory budgeting, \citeauthor{FPPV21} explore the problem of designing strategyproof mechanisms that satisfy proportionality.  One of their key results (Proposition 1) applies to the facility location problem and shows that there is a unique anonymous, continuous, strategyproof and proportional mechanism, which is called the Uniform Phantom  mechanism. Like our paper, \citeauthor{FPPV21}'s (\citeyear{FPPV21}) setting assumes that agents have single-peaked and symmetric preference. \citet{JLPV21}  provide a similar characterization of the Uniform Phantom mechanism in the setting where agents have single-peaked (and possibly asymmetric) preferences. 
Our paper differs in focus and provides a broader treatment of issues of fairness and strategyproofness in facility location problems; for example, we  characterize a larger family of strategyproof mechanisms that satisfy the weaker fairness axiom of IFS. In addition, one of our results strengthens \citeauthor{FPPV21}'s Proposition 1  by showing that the anonymity axiom is redundant in their characterization. We also provide an alternative characterization of the Uniform Phantom mechanism as the equilibrium outcome of any continuous, strictly monotonic, and UFS mechanism. 

Finally, we note that in more general mechanism design problems, ``fairness'' is often explored in a relatively minimal manner. For example, \citet{Spru91} interprets a mechanism to be fair if it satisfies anonymity and envy-freeness, and \citet{Moul17} interprets a mechanism to be fair if it satisfies anonymity, envy-freeness, and a status-quo participation constraint. These minimal notions of fairness have persisted because of various impossibility results in the literature. For example, Theorem 3 of \citet{BoJo83} shows that, for the multi-dimensional facility location problem with not necessarily separable preferences, there is no  strategyproof, unanimity-respecting, and anonymous mechanism~\citep[see also][]{Laff80}. Like~\citet{Spru91} and~\citet{Moul17},  the unidimensional facility location problem  that we study escapes these impossibility results. Our paper contributes a complementary set of fairness axioms that go beyond the basic requirement of anonymity and connect to the notion of proportional representation. We do not consider envy-freeness since, in the context of the facility location problem, it is trivially satisfied by any facility location~\citep[see, e.g., Section 8.1 of][]{Moul17}. The status-quo participation constraint explored by  \citet{Moul17} requires that an agent weakly prefers the mechanism's outcome to some status-quo outcome. This is distinct but has a similar flavor to our IFS axiom, which is one of our weakest fairness axioms. The IFS axiom requires that the facility location is not located too far from any agent. When reframed in terms of utility, IFS enforces a minimum utility guarantee for all agents, which could be viewed as an outside option.

\paragraph{Approximate mechanism design.} The final section of our paper explores the performance of strategyproof and fair mechanisms with respect to maximizing social (or utilitarian) welfare  and minimizing total cost. Adopting the approximation ratio approach of~\citet{NiRo01a} and~\citet{PrTe13}, we measure the performance of these mechanisms by their worst-case performance over the domain of possible preferences profiles relative to the {welfare-optimal}  mechanism and the total cost-optimal mechanism. This is a common approach in the economics and computation literature~\citep[see, e.g.,][]{ACLL+20,ACLP20,FFG16,NiRo01a}. For our main fairness axioms of Proportionality, IFS, UFS and PF, we identify the best performing strategyproof and fair mechanism.  {In particular, we find that the Uniform Phantom mechanism has the best welfare approximation ratio among all mechanisms satisfying UFS (including non-strategyproof mechanisms).  In the participatory budgeting setting, \citet*{CCP22} show a related result: when there are only 2 projects, the Uniform Phantom mechanism achieves the best cost approximation ratio among all strategyproof mechanisms.}

  
\section{Model.}\label{sec:model}


 Let $N=\{1, \ldots, n\}$ be a set of agents with $n\ge 2$ and let $X:=[0,L]$ be the domain of locations.  The restriction to $X=[0,L]$ is without loss of generality for any closed interval of real numbers.   The restriction of locations to an interval  is common in the literature~\citep[see, e.g., seminal works][]{BaJa94,Chin97,MaMo11a}; it is also well-suited to many real-world problems, such as deciding the level of provision of a public good, which is naturally constrained to be between zero and  the total available budget.  A  \emph{mechanism} is a mapping $f\ : \ X^n\rightarrow X$ from a (reported) location profile $\hat{\boldsymbol{x}}=(\hat{x}_1, \ldots, \hat{x}_n)\in X^n$ to a facility location $y\in X$.     Let $U$ be the set of all symmetrically single-peaked utility functions on $X$. 
 That is, given a  function  $u\in U$, there exists a unique ``peak''  location $x\in X$  that maximizes $u$ and $u$ is symmetrically decreasing around $x$~\citep[see, e.g.,][]{BoJo83,PvdsS92,KHS98}. 
Each agent $i$ has utility function $u_i\in U$. We interpret agent $i$'s peak  $x_i$ as agent $i$'s location. Each agent's  utility function $u_i$ (and, hence, location $x_i$) is privately known to the agent and is not assumed as an input into the mechanism. We refer to agent $i$'s cost as the distance between their location and the facility's location,  i.e., $d(y, x_i)=|y-x_i|$.  Notice that  $u_i$  is decreasing in agent $i$'s cost: $d(y, x_i)$.

A widely accepted---albeit minimal---fairness principle is that a mechanism should not depend on the agents' labels. This is referred to as anonymity. 

\medskip{}

\begin{definition}[Anonymous]\label{def: anon}
A mechanism $f$ is \emph{anonymous} if,  for every  location profile  $\hat{\boldsymbol{x}}$  and every bijection $\sigma \ : \ N \rightarrow N$,
$$f(\hat{\boldsymbol{x}}_{\sigma})=f(\hat{\boldsymbol{x}}),$$
 where $\hat{\boldsymbol{x}_{\sigma}}:=(\hat{x}_{\sigma(1)},  \hat{x}_{\sigma(2)}, \ldots,  \hat{x}_{\sigma(n)}).$
\end{definition}

\medskip{}
 
Given a location profile $\hat{\boldsymbol{x}}$, a facility location $f(\hat{\boldsymbol{x}})=y$ is said to be \emph{Pareto optimal} if there is no other facility location $y'$ such that for all $i\in N$, $u_i(y')\ge u_i(y)$, with strict inequality holding for at least one agent.  A mechanism $f$ is said to be \emph{Pareto efficient} if, for every location profile $\hat{\boldsymbol{x}}$, the facility location $f(\hat{\boldsymbol{x}})$ is Pareto optimal. In our setting, Pareto optimality is equivalent to requiring that $y\in [\min_{i\in N}\hat{x}_i, \max_{i\in N} \hat{x}_i]$.

 We are interested in mechanisms that are ``strategyproof'', i.e., the mechanism never incentivizes an agent to misreport their location. Before providing a formal definition, we introduce some notation. Given a profile of locations (or reported locations) $\boldsymbol{x}'$,  the profile $(\boldsymbol{x}_{-i}', x_i'')$ denotes the profile obtained by swapping $x_i'$ with $x_i''$ and leaving all other agent locations (or reports) unchanged.

\medskip{}

\begin{definition}[Strategyproof]\label{def: SP}
A mechanism $f$ is  \emph{strategyproof} if, for every agent $i\in N$ with peak location $x_i$, we have that for every $x_i'$ and ${\boldsymbol{x}}_{-i}'$,
%
$$ u_i(f( {\boldsymbol{x}}_{-i}', x_i))\ge u_i(f({\boldsymbol{x}}_{-i}', x_i')). $$
\end{definition}

 \medskip{}
 

 Our focus is on characterizing mechanisms that are strategyproof, Pareto efficient, and anonymous, while also satisfying additional notions of proportionality-based fairness (to be introduced in Section~\ref{subsection: fairness axioms}). To assist with interpretation, our model assumes that agents' utilities are symmetric and single-peaked. However, our characterization results in Section~\ref{section: fair and SP} do not require that agents' utilities be symmetric about their peak. This   follows from  Corollary 2 of~\cite{MaMo11a}, which says that, when agents have single-peaked preferences, the  set of strategyproof, Pareto efficient, and anonymous mechanism is unchanged whether or not agents have symmetric preferences. 
 
Omitted proofs  appear in the Appendix.


 \section{Proportionality-based fairness.}
\label{subsection: fairness axioms}


We now introduce a hierarchy of proportionality-based fairness  axioms. The first three axioms have previously been proposed in the literature; the fourth axiom, Proportional Fairness, is a new concept that we propose.    We formulate our fairness axioms in terms of the cost  (or distance) function $d(y,x_i)$.  In Section~\ref{section: defense}, we provide further motivation for our axioms with a running example; we also provide a discussion and justification for the distance-based formulation of our axioms.

The first axiom, \textbf{Individual Fair Share (IFS)}, requires that the facility location imposes a cost on each agent of no more than $L(1-\frac{1}{n})$. In other words, each agent is entitled to avoid $1/n$-th of the  maximum possible cost. {In the context of cake-cutting, IFS coincides with the axiom of \citet{Stei48a} commonly known as proportionality. It also appears as the ``Fair Welfare Share'' axiom in the context of participatory budgeting, as defined by \citet{BMS05}.}

%
%
%
%

\medskip{}

\begin{definition}[Individual Fair Share (IFS)]\label{def: IFS}
Given a profile of locations $\boldsymbol{x}$,  a facility location $y$ satisfies \emph{Individual Fair Share (IFS)} if each agent has cost of at most $L(1-\frac{1}{n})$, i.e., for all $i \in N$,
$$d(y, x_i)\le L(1-1/n).$$
\end{definition}

\medskip{}

The second axiom, \textbf{Unanimous Fair Share (UFS)}, is a strengthening of IFS. UFS considers all subsets of agents that share the same location; let $S\subseteq N$ be such a subset of agents. UFS requires that the facility location imposes a cost on each agent in $S$ of no more than $L(1-\frac{|S|}{n})$. In other words, a subset of agents $S$ is entitled to avoid  $|S|/n$-th of the maximum possible cost.  In the context of participatory budgeting, UFS appears in~\citet{ABM19}.
 
\medskip{}
 \begin{definition}[Unanimous  Fair Share (UFS)]
Given a profile of locations $\boldsymbol{x}$ such that a subset of $S\subseteq N$ agents share the same location,  a facility location $y$ satisfies \emph{Unanimous Fair Share (UFS)} if for all $i \in S$,
$$d(y, x_i)\le L(1-\frac{|S|}{n}).$$
\end{definition}
\medskip{}

The third axiom \textbf{Proportionality} requires that,  if all agents are located at ``extreme'' locations (i.e., $0$ or $L$),  the facility is  located at the average of the agents' locations. \citet{FPPV21} focus on this axiom in a participatory budgeting setting. 
\medskip{}
\begin{definition}[Proportionality]\label{def: F-prop}
Given a profile of locations $\boldsymbol{x}$ such that $x_i\in \{0,L\}$ for all $i\in N$, a facility location $y$ satisfies \emph{Proportionality} if $y=L\frac{|i\in N \ : \ x_i=L|}{n}.$
\end{definition}
\medskip{}

Finally, we propose a new fairness concept called \textbf{Proportional Fairness (PF)}. PF considers all subsets of agents. Given a subset of agents $S\subset N$, PF requires that the facility location imposes a cost on each agent in $S$ that depends  on both the size of the group, $|S|$, and how closely the agents in $S$ are clustered. The idea behind the concept is similar in spirit to proportional representation axioms in voting which require that if a subset of agents is large enough and the agents in the subset have ``similar"  preferences, then the agents in the subset deserve an appropriate level of representation~\citep[see, e.g.,][]{ABC+16a,AzLe20,AzLe22,Dumm84a,SFEL+17}.
\medskip{}
\begin{definition}[Proportional Fairness (PF)]\label{def:PF}
	Given a profile of locations $\boldsymbol{x}$, a facility location $y$ satisfies \emph{Proportional Fairness (PF)} if, for any subset of agents $S\subseteq N$ within a range of distance $r:=\max_{i\in S}\{x_i\}-\min_{i\in S}\{x_i\}$, the agents in $S$  have at most $L(1-\frac{|S|}{n})-r$ cost, i.e. for all $i\in S$, 
	$$d(y, x_i)\le L(1-\frac{|S|}{n})+r .$$
	\end{definition}
	\medskip{}
	
	In the definition of PF,  given a group $S$, $r$ is non-negative and equals zero if and only if all agents in $S$ share the same location. Hence, PF implies UFS. For any $r$ that is larger, the corresponding fairness concept is weaker. For any $r$ that is smaller, there may not exist any outcome that satisfies the corresponding definition.

 A natural---albeit weak---notion of fairness is called \textbf{Unanimity}. It requires that, if all agents are unanimous in their most preferred location, then the facility is  located at this same location. Notice that Pareto optimality implies unanimity.

\medskip{}
\begin{definition}[Unanimity]\label{def: unanimity}
Given a profile of locations $\boldsymbol{x}$ such that $x_i=c$ for some $c\in X$ and for all $i\in N$, a facility location $y$ satisfies \emph{unanimity}  if $y=c$.
\end{definition}
\medskip{}

Proposition~\ref{prop: UFS implies unanimity} establishes the logical connection between the  fairness axioms. Figure~\ref{Figure: relation} provides an illustration of proposition. PF is the strongest fairness notion:  it implies all of the other axioms (UFS, IFS, Proportionality, and Unanimity). The next strongest axiom is UFS: it implies IFS, proportionality, and unanimity. There is no relationship between proportionality, IFS, and unanimity; however,  as will be shown, they are compatible with each other.   
\medskip{}
\begin{proposition}[A hierarchy of axioms]\label{prop: UFS implies unanimity}
\
\begin{enumerate}
\item UFS implies proportionality, IFS, and unanimity
\item PF implies UFS
\end{enumerate}
All of the above relations are strict; there is no logical relation between proportionality, IFS, and unanimity.  Figure~\ref{Figure: relation} provides an illustration. 
\end{proposition}
\medskip{}

\subsection{Discussion of our fairness axioms.}\label{section: defense}

\paragraph{Motivation.} In addition to normative appeals to fairness,  all of the proportionality-based fairness axioms   can be motivated by concerns for the sustainability and practicality of collective decision making. As a running example, suppose that the facility location corresponds to the level of provision of a public good. The total budget is $L$, which each agent contributed equally to (i.e., $L/n$), and any unspent budget is saved for a future year. Agents have (possibly different) preferences over the tradeoff between current spending on the public good and future savings. Each agent's peak location corresponds to their ideal provision of the public good (the complement of this is their ideal provision of savings). An intuitive requirement is that each agent---having contributed $1/n$-th of the total budget---should  be able to avoid the total budget (respectively, none of the budget) being spent if their ideal provision of the public good is to spend nothing (resp., spend all of the budget). Indeed, one could imagine that an outcome that does not abide by this requirement would be unsustainable and impractical in reality:  the agent could  withdraw their contribution from the budget and independently not fund (resp., fund) a $1/n$-th    share of the public good. This requirement is reminiscent of stability solution concepts in cooperative game theory, such as  the ``core''~\citep[see, e.g., ][]{Scar67}.  The IFS axiom extends this requirement to agents that---not only have an ``extreme'' ideal provision of the public good (i.e., spending all or nothing)---but also those that have ideal provisions   close to these extremes.
However, building on these same ideas, it might expected that if a single agent can control $1/n$-th of the total budget, then a group of like-minded agents, say of size $|S|$ and who all share a common ideal provision of the public good, can control  $|S|/n$-th of the total budget. The UFS axiom strengthens the IFS axiom by incorporating this ``group'' consideration into the decision-making process. The Proportionality axiom is similar; however, it only applies to instances where agents can be partitioned into two groups that have extreme ideal provisions of the public good (i.e., spending all or nothing). The unanimity axiom is a special case of the UFS axiom.  Finally, the PF axiom relaxes the notion of a ``group'' of agents that is implicit in the  UFS (and proportionality) axioms. Intuitively, a group of $|S|$ agents might be able to control $|S|/n$-th of the total budget even if they are not perfectly unified in their ideal provision of the public good. It may simply be enough that the group members have ideal provisions that are ``close enough''---in which case, they can still control $|S|/n$-th of the total budget to achieve mutually beneficial outcomes. The PF axiom incorporates this more flexible notion of a ``group'' and formalizes what such a group can achieve by controlling $|S|/n$-th of the  budget. Intuitively, the more closely aligned a group is in their ideal provision (i.e., a smaller value of $r$ in Definition~\ref{def:PF}), the more precisely they can use their control of the budget to achieve an outcome close to their ideal provision.

\paragraph{Distance-based formulation of our axioms.} 
We  formulated our axioms in terms of Euclidean distance. Because agents have symmetric and single-peaked utility functions, an agent's utility is strictly decreasing in their distance from the facility. Therefore, our axioms have direct implications for agents' utilities but, importantly, do not  correspond to  a precise utility guarantee. Our approach is more general than simply assuming a specific functional form for all agents' utility functions (as is  sometimes done in the facility location literature~\citep[see, e.g.][]{AnDe18,ACLP20,DeFRV23}) and then constructing axioms that depend on the assumed functional form. 

Our approach is also motivated by practical concerns.  In our setting, obtaining precise utility guarantees   requires the mechanism to elicit information about each agent's entire utility function (i.e.,  not only reporting their peak location). Yet it is well known in the literature that strategyproofness is incompatible with eliciting information beyond an agent's peak location~\citep[see, e.g.,][]{BaJa94,Weym08}. Therefore, in the pursuit of strategyproof and proportionally fair mechanisms, we are forced to act behind a \emph{veil of ignorance}. It seems reasonable that a ``fair” outcome should, at minimum: 
\begin{enumerate}
\item impose conditions on the ``closeness'' between agents' peaks and the facility location because this has direct implications on agents' utilities;
\item the measure of closeness should be symmetric; 
\item the measure of closeness should be anonymous. 
\end{enumerate}
These points imply that a single benchmark distance metric should be applied for each agent. We adopt the standard Euclidean distance for our axioms (IFS, UFS, PF), i.e., $d(y,x_i)$ equals $|y-x_i|$; this has desirable and natural features. For example, suppose  $n=2$ with one agent located at $0$ and the other at $L$. The absolute value $|y-x_i|$  is the only metric that requires the facility to be located at exactly $\frac{L}{2}$ via the IFS condition $d(y,x_i)\le  L(1-1/n)$ (the same is true for the UFS and PF conditions). Lower powers of $|y-x_i|$ could be considered (i.e., $|y-x_i|^p$ for $0<p<1$) but this leads to non-existence. Higher powers could be considered (i.e., $|y-x_i|^p$ for $p>1$) but this leads to the possibility of ``fair” outcomes that asymmetrically favor one agent over the other. To see this, suppose $p=2$ and $L=1$. The IFS condition when $n=2$, one agent is located at 0, and the other at $1$, becomes $y^2\le  \frac{1}{2}$ and  $(1-y)^2\le  \frac{1}{2}$. This IFS condition is equivalent to requiring $y \in [\frac{1}{2}(2-\sqrt{2}), \frac{1}{\sqrt{2}}]$, which admits asymmetric  solutions such as $y=0.7$.  


\paragraph{Restrictions on   agents' peak locations.} Another potential concern is that our distance-based axioms implicitly assume that each agent's peak location is contained in the interval $X=[0,L]$.
The fact that the facility must be located in a (fixed and known) closed interval of the real line and each agent's peak location (and reported location) are constrained to be in this interval are common assumptions in the literature~\citep[see, e.g.,][]{BaJa94,Chin97,MaMo11a}. The assumptions are also appropriate for important settings of interest, such as the provision of a public good. Our model adopts these common assumptions, and our proportionality-based fairness axioms build on these same assumptions. We note, however, that our axioms   and results can be modified to a setting where the mechanism must locate the facility in the interval $X=[0,L]$ but agents' peak locations may lie on $\mathbb{R}$ (in particular, beyond the interval $X$) and may also report locations beyond the interval. The set of mechanisms that we focus on are essentially unaffected by this modification.  To be slightly more precise,  the mechanisms that we focus on can be extended to this modified setting via the following procedure: if an  agent $i$ reports $\hat{x}_i<0$ (resp., $\hat{x}_i>L$), then the mechanism input for agent $i$ becomes 0 (resp., $L$); if an  agent $i$ reports $\hat{x}_i\in [0,L]$, then the mechanism input for agent $i$ is simply $\hat{x}_i$. It is straightforward to see that this modified setting does not generate any additional strategyproof, anonymous, and Pareto efficient mechanisms that also guarantee a facility location in $[0,L]$.     Our results can then be recovered with appropriately modified versions of our axioms that replace the distance function, $d(y, x_i)$, with 
$$\tilde{d}(y,x_i)=\begin{cases}
d(y, 0) &\text{if $x_i<0$,}\\
d(y, x_i) & \text{if $x_i\in [0,L]$,}\\
d(y, L)&\text{if $x_i>L$.} 
\end{cases} $$

\section{Strategyproof and Proportionally Fair Mechanisms.}\label{section: fair and SP}


We begin by {reviewing} some prominent mechanisms from the literature. The \textbf{median mechanism} $f_{\text{med}}$  places the facility at the median location (i.e., the $\floor{n/2}$-th location when locations are placed in increasing order). The median mechanism is sometimes referred to as the utilitarian mechanism since it places the facility at a location that   minimizes the sum of agent costs.

The \textbf{midpoint mechanism} $f_{\text{mid}}$ places the facility at  the midpoint of the leftmost and rightmost agents, i.e., 
\begin{align}\label{eq: mid mechanism L}
f_{\text{mid}}(\boldsymbol{x})=\frac{1}{2}\left(\min_{i\in N}x_i+\max_{i\in N}x_i\right).
\end{align}
The midpoint mechanism is sometimes referred to as the egalitarian mechanism since it minimizes the maximum agent cost.

 A \textbf{Nash mechanism} places the facility at a location that maximizes the product of agent utilities: $\prod_{i\in N} u_i(y)$. In our model, agents' utility functions $u_i$ are not reported (agents only report locations); furthermore, the Nash mechanism is only well-defined when $u_i(y)$ is non-negative for facility locations $y\in X$. Therefore, to define the Nash mechanism in our setting---and using the benefit of hindsight---we adopt the following form:  a Nash mechanism $f_{\text{Nash}}$ locates the facility at  
\begin{align}\label{eq: nash mechanism L}
f_{\text{Nash}}(\boldsymbol{x})= \arg \max_{y\in [0,1]} \prod_{i\in N} \big(L-d(y, x_i)\big).
\end{align}
The formulation above says that the Nash mechanism operates upon the (no necessarily true) assumption that all agents have a  utility function of the form $u_i(y)=L-d(y, x_i)$. When each agent's true utility function is $u_i(y)=L-d(y, x_i)$, the Nash mechanism is described  by \citet[p. 80]{Moul03} as achieving a ``sensible compromise between utilitarianism and egalitarianism."
 
  \paragraph{Incompatibility results.} 
 All of the above mechanisms either fail to provide fair outcomes (per the axioms in Section~\ref{subsection: fairness axioms}) or fail to be strategyproof. The median mechanism fails Proportionality and IFS; however, it is strategyproof and satisfies unanimity. The midpoint mechanism---often heralded as a hallmark of fairness---fails to satisfy many of Section~\ref{subsection: fairness axioms}'s proportionality-based fairness axioms; it only satisfies the weakest axioms: IFS and unanimity. Furthermore, the midpoint mechanism is not strategyproof. Finally, the Nash mechanism, as formulated in (\ref{eq: nash mechanism L}), obtains the strongest axiom of proportional fairness, PF---and, hence, satisfies the other fairness axioms: UFS, Proportionality, IFS, and unanimity. However,   the Nash mechanism is not strategyproof~\citep{LAW21}. Proposition~\ref{prop: Egal Nash} summarizes these results.  
 \medskip{}
 \begin{proposition}[Review of existing mechanisms]\label{prop: Egal Nash}
\
\begin{enumerate}
\item The median mechanism satisfies unanimity and strategyproofness, but does not satisfy IFS, PF, UFS nor Proportionality.
\item The midpoint mechanism satisfies IFS and unanimity, but it is not strategyproof. The midpoint mechanism does not satisfy PF, UFS, nor Proportionality. 
\item  The Nash mechanism satisfies PF, but it is not strategyproof.  
\end{enumerate}
\end{proposition}
\medskip{}

\subsection{Characterization of IFS and strategyproof mechanisms.}\label{section:char1}

We now characterize the family of strategyproof and IFS mechanisms. Our characterization leverages the class of Phantom mechanisms introduced by~\citet{Moul80}~\cite[see also][]{BoJo83}. Although both~\citet{Moul80} and~\citet{BoJo83} deal with a setting where agents' locations are in $\mathbb{R}$ rather than $[0,L]$,  their results extend naturally~\citep[see, e.g.,][]{MaMo11a}.   Intuitively, Phantom mechanisms can be understood as locating the facility at the median of $2n-1$ reports, where $n$ reports correspond to the agents' reports and $n-1$ reports are fixed (and pre-determined) at locations $p_1, \ldots, p_{n-1}$. The fixed reports are referred to as ``phantom'' locations.
  \medskip{}
\begin{definition}[Phantom Mechanisms]\label{def: phantom}
Given $\boldsymbol{x}\in X$ and $n-1$ values $0\le p_1\le \cdots \le p_{n-1}\le L$, a Phantom mechanism  locates the facility at $\text{Median}\{x_1, \ldots, x_n, p_1, \ldots, p_{n-1}\}.$
\end{definition}

\medskip{}

The family of Phantom mechanisms is broad and captures many well-known mechanisms. To build intuition, we provide some examples below.
\begin{enumerate}[leftmargin=*,labelindent=10pt,label=  \arabic*.]
\item The classic median mechanism is obtained by locating $\floor{(n-1)/2}$ phantoms at $0$ and $\ceil{(n-1)/2}$ phantoms at $L$.
\item The ``Maximum'' (resp., ``Minimum'') mechanism, which locates the facility at the maximum (resp., minimum) agent location, is obtained by locating all the phantoms at $L$ (resp., $0$).
\item The ``Moderate$-\frac{L}{2}$'' mechanism, which locates the facility at the minimum (resp., maximum) agent reported location when all agents report above (resp., below) $L/2$ and otherwise (i.e., when some agent(s) report either side of $L/2$) the facility is located at $L/2$. This mechanism is obtained by locating all the phantoms at $L/2$.
\end{enumerate}
\medskip{}

On the other hand, mechanisms such as the midpoint mechanism  (\ref{eq: mid mechanism L}) and the Nash mechanism (\ref{eq: nash mechanism L}) from Section~\ref{section: fair and SP} do not belong to the family of Phantom mechanisms. Similarly, the ``Average'' mechanism, which locates the facility at the average of all agents' reports, is not a Phantom mechanism. Given 6 agents with $L=1$ and location profile $\boldsymbol{x}=(0,0,0,0,0.8,1)$, Figure~\ref{figure:rules} provides an illustration of these mechanisms (and also other mechanisms that will be  defined later). Each agent's location is depicted by an `x' mark; each mechanism's facility location is depicted by a  \textbullet \ (with label directly above). Further details are provided in the figure caption.

   \begin{figure}[h!]
 	  	 \begin{center}  
   		      	             \begin{tikzpicture}[scale=1]
   			   	      	                 \centering
   	   	      	                 \draw[-] (0,0) -- (12,0);
								 
   								   \draw[-] (0,0) -- (0,0.25);
   								    \draw[-] (12,0) -- (12,0.25);
   									    \draw[-] (2,0) -- (2,0.25);
   										 \draw[-] (4,0) -- (4,0.25);
   										  \draw[-] (6,0) -- (6,0.25);
    \draw[-] (8,0) -- (8,0.25);
     \draw[-] (10,0) -- (10,0.25);
  
     \draw (0,-.4) node(c){\small $0$};
       \draw (2,-.4) node(c){\small $1/6$};
   	    \draw (4,-.4) node(c){\small $2/6$};
   		    \draw (6,-.4) node(c){\small $3/6$};
   			 \draw (8,-.4) node(c){\small $4/6$};
   			 	 \draw (10,-.4) node(c){\small $5/6$};
       \draw (12,-.4) node(c){\small $1$};
	
   	  \draw (0,0.6) node(c){\small $\text{x}$};
   	  \draw (0,0.8) node(c){\small $\text{x}$};
   	  	  \draw (0,1) node(c){\small $\text{x}$};
   		  	  	  \draw (0,1.2) node(c){\small $\text{x}$};
				  
   				   \draw (12,1) node(c){\small $\text{x}$};
   				   	   \draw (9.6,1) node(c){\small $\text{x}$};
 					   		\draw (3.4,2.5)node(c){\small $y_{\text{Nash}}$};
   					   	   \draw (4,1.5) node(c){\small $y_{\text{Unif}}$};
   						    \draw (2,2) node(c){\small $y_{\text{CM}}$};
   						   \draw (0,2) node(c){\small $y_{\text{med}}$};
   						    \draw (6,2) node(c){\small $y_{\text{mid}}$};
   							\draw (3.6,2) node(c){\small $y_{\text{avg}}$};
   							\draw (0,0) node(c)[circle,fill,inner sep=1.5pt]{};
   							\draw (2,0) node(c)[circle,fill,inner sep=1.5pt]{};
   							\draw (3.4,0) node(c)[circle,fill,inner sep=1.5pt]{};
   							\draw (3.6,0) node(c)[circle,fill,inner sep=1.5pt]{};
   							\draw (6,0) node(c)[circle,fill,inner sep=1.5pt]{};
   							\draw (4,0) node(c)[circle,fill,inner sep=1.5pt]{};
   			   	      	  \end{tikzpicture}
   			   	       	\end{center}
   			   	      	 \caption{Facility location problem on the $[0,1]$ domain with $n=6$ agents, with location profile $(0,0,0,0, 0.8,1)$ represented by x. The facility locations (represented by \textbullet) correspond to the: Median mechanism, $y_{\text{med}}=0$; Constrained Median mechanism,  $y_{\text{CM}}=\frac{1}{6}$; Nash mechanism, $y_{\text{Nash}}\approx 0.284$; Average mechanism, $y_{\text{avg}}=0.3$; Uniform Phantom mechanism, $y_{\text{Unif}}=\frac{2}{6}$; and Midpoint mechanism, $y_{\text{mid}}=\frac{3}{6}$.}
   			   	      	\label{figure:rules}
   			   	      	\end{figure}	

In our setting, the family of Phantom mechanisms are known to characterize all strategyproof, anonymous, and Pareto efficient mechanisms~\citep[Corollary 2 of][]{MaMo11a}.
 This characterization of Phantom mechanisms forms the foundation of our characterization results.

Theorem~\ref{char: IFS and SP} says that the family of IFS, strategyproof, anonymous, and unanimous mechanisms are characterized by the subfamily of Phantom mechanisms that have their phantom locations contained in the interval $[\frac{1}{n}, L-\frac{1}{n}]$. Intuitively, when the facility is located in the interval  $[\frac{1}{n}, L-\frac{1}{n}]$, IFS is satisfied regardless of the agents' locations. The restricted class of Phantom mechanisms in Theorem~\ref{char: IFS and SP} satisfies IFS by preventing the facility from being located at an ``extreme'' point (i.e., beyond the interval $[\frac{1}{n}, L-\frac{1}{n}]$) unless all agents are   located close together and at a common extreme point.

 \medskip{}
\begin{theorem}[Characterization: IFS, unanimous, anonymous, and strategyproof]\label{char: IFS and SP}
A mechanism is strategyproof, unanimous, anonymous and satisfies IFS if and only if it is a Phantom mechanism with $n-1$ phantoms all contained in the interval $[\frac{1}{n}, L-\frac{1}{n}]$.
\end{theorem}
\medskip{}

\begin{proof}
We start with the backwards direction. Let $f$ be a Phantom mechanism with the $n-1$ phantoms contained in $[\frac{L}{n}, L(1-\frac{1}{n})]$. First note that $f$ is strategyproof because all Phantom mechanisms are strategyproof \citep[see, e.g., Corollary 2 of][]{MaMo11a}. Furthermore, it is  immediate from the Phantom mechanism definition (Definition~\ref{def: phantom}) that $f$ satisfies unanimity.  It remains to show that $f$ satisfies IFS. To see this, notice that the  facility is located above (resp., below) both of the endpoints of the interval $[\frac{L}{n}, L(1-\frac{1}{n})]$ if and only if all agents are located above (resp.,  below) of  the interval. Therefore, in such cases, the facility is located within a distance of $\frac{L}{n}$ of all agents. Otherwise, the facility is located  within the interval and the largest possible cost is $L(1-\frac{1}{n})$, as required.

We now prove the forward direction. Let $f$ be a mechanism that is strategyproof,  unanimous, anonymous, and satisfies IFS. \citeauthor{BoJo83}'s (\citeyear{BoJo83}) Lemma~3 says that any strategyproof and unanimous mechanism is Pareto efficient.   Hence,  $f$ is strategyproof, IFS, unanimous, anonymous, and Pareto efficient. We now apply Corollary~2 of~\citet{MaMo11a}, which says that a mechanism is strategyproof, anonymous, and Pareto efficient if and only if it is a Phantom mechanism (Definition~\ref{def: phantom}). We now show that $p_j\in [\frac{L}{n}, L(1-\frac{1}{n})]$ for all $j\in \{1, \ldots, n-1\}$. For the sake of a contradiction, suppose $p_1<\frac{L}{n}$ (the case of $p_{n-1}>L(1-\frac{1}{n})$ is dealt with similarly and, hence, is omitted). If $n-1$ agents are located at $0$ and the remaining  agent is located at $L$, then the facility must be located at $p_1<\frac{L}{n}$. But then the agent at location $L$ experiences cost strictly greater than $L(1-\frac{1}{n})$---a contradiction of IFS. Therefore, $p_j\in [\frac{L}{n}, L(1-\frac{1}{n})]$ for all $j\in \{1, \ldots, n-1\}$, as required.
\end{proof}
\medskip{}
 Theorem~\ref{char: IFS and SP}  is ``tight'' in the following sense: if any one of the requirements in  Theorem~\ref{char: IFS and SP} (i.e., strategyproofness, unanimity, anonymity, and IFS) is removed, then the theorem fails to hold. In Appendix~\ref{section:prop: ess char}, for each smaller set of requirements, we identify a mechanism that  satisfies them and does not belong to the family of mechanisms described in  Theorem~\ref{char: IFS and SP}.


\subsection{Characterization of PF, UFS, Proportional, and strategyproof mechanisms.}\label{section:char2} 
 
 
%

 We now show that strategyproofness and PF are compatible and can be achieved via the ``Uniform Phantom'' mechanism. By Proposition~\ref{prop: UFS implies unanimity} this also implies that UFS and, hence, proportionality, IFS, and unanimity can be attained simultaneously.  The Uniform Phantom mechanism is obtained from the general class of Phantom mechanisms  (Definition~\ref{def: phantom}) by locating the $(n-1)$ phantoms at  $\frac{jL}{n}$ for $j=1, \ldots, n-1$. Figure~\ref{figure:rules} provides an illustration of the mechanism. This mechanism is the focus of~\citet{FPPV21}; later we   provide a discussion of the similarities and differences between our results and those of~\citeauthor{FPPV21}.
\medskip{}
\begin{definition}[Uniform Phantom mechanism]\label{def: Uniform Phantom}
Given $\boldsymbol{x}\in X$, the Uniform Phantom mechanism $f_{\text{Unif}}$ locates the facility at
$$\text{Median}\{x_1, \ldots, x_n, \frac{L}{n}, \frac{2L}{n}, \ldots, \frac{(n-1)L}{n}\}.$$
\end{definition}

  		\medskip{}
			   	   
It is immediate that the Uniform Phantom mechanism is strategyproof since it belongs to the family of Phantom mechanisms (Definition~\ref{def: phantom}). However, in addition to strategyproofness, Proposition~\ref{prop: uniform is SP and PF} says that the Uniform Phantom mechanism satisfies PF. Intuitively, the Uniform Phantom mechanism locates the facility at the $n$-th location of the  $2n-1$ phantom and agent locations. Given the phantom locations, for every $L/n$ units of distance, there is at least one phantom. Therefore, for any set of agents $S$, the  distance between the most extreme agents in $S$ and the facility is at most $L\frac{n-|S|}{n}$ and, hence,  the distance between any agent in $S$ and the facility is at most $L\frac{n-|S|}{n}+r$, where $r$ is the range of the agents in $S$.

 \medskip{}
\begin{proposition}[Uniform Phantom mechanism properties]\label{prop: uniform is SP and PF} \ \\ 
The Uniform Phantom mechanism is strategyproof and satisfies PF. Thus, it also satisfies UFS, IFS, proportionality, and unanimity.
\end{proposition}
\medskip{}


A natural question is whether there exist other strategyproof mechanisms satisfying UFS or proportionality and unanimity. It turns out that there are not: Theorem~\ref{thm: no need anon Freeman} says that the Uniform Phantom mechanism is the only strategyproof mechanism that is  proportional and unanimous. A key challenge in the theorem is that anonymity is not supposed and hence, the well-known characterization of Phantom mechanisms cannot be immediately applied. In the appendix, we prove an auxiliary lemma that says anonymity is implied by strategyproofness, unanimity, and proportionality. With this in hand, the Phantom mechanism characterization can be utilized. Proportionality then implies the (unique) locations of the $n-1$ phantoms. This is because of two observations. First, proportionality requires that, for any $k=1, \ldots, n-1$, when $k$ agents are located at $L$ and $n-k$ agents at $0$, the facility is located at $kL/n$.  Second, for such a profile of locations, any Phantom mechanism will locate the facility at the $k$th phantom. Therefore,   the phantoms  must be located at $\frac{kL}{n}$ for $k=1, \ldots, n-1$.

\medskip{}
\begin{theorem}[Characterization: proportional, unanimous, and strategyproof]\label{thm: no need anon Freeman} \ \\
A mechanism satisfies strategyproofness,  unanimity, and proportionality   if and only if it is the Uniform Phantom mechanism. 
\end{theorem}
\medskip{}

\begin{proof}
The backward direction follows immediately from Proposition~\ref{prop: uniform is SP and PF} and Proposition~\ref{prop: UFS implies unanimity}. It remains to prove the forward direction. Suppose $f$ is strategyproof and satisfies proportionality and unanimity. We utilize an auxiliary lemma  (Lemma~\ref{lemma: SP, unan, prop implies anon}), which says that any strategyproof, unanimous, and proportional mechanism must be anonymous.  The proof of Lemma~\ref{lemma: SP, unan, prop implies anon} is quite involved and is proven in Appendix~\ref{section:lemma: SP, unan, prop implies anon}. Given Lemma~\ref{lemma: SP, unan, prop implies anon}, we apply  \citeauthor{BoJo83}'s (\citeyear{BoJo83}) Lemma~3 (i.e., any strategyproof and unanimous mechanism is Pareto efficient). This tells us that $f$ must also be anonymous and Pareto efficient. We now apply Corollary~2 of~\citet{MaMo11a}, which says that a mechanism is strategyproof, anonymous, and  Pareto  efficient if and only if it is a Phantom mechanism (Definition~\ref{def: phantom}). We now show that $p_j=\frac{jL}{n}$ for all $j\in \{1, \ldots, n-1\}$. To see this, take arbitrary $j\in \{1, \ldots, n-1\}$, and let $\boldsymbol{x}$ be a profile of locations such that there are $j$ agents at $L$ and $n-j$ agents at $0$. By definition of the Uniform Phantom mechanism, $f(\boldsymbol{x})=p_j$. But proportionality  requires that $f(\boldsymbol{x})=\frac{jL}{n}$; hence, $p_j=\frac{jL}{n}$. This completes the proof. 
\end{proof}
\medskip{}
Combining Proposition~\ref{prop: UFS implies unanimity} and Proposition~\ref{prop: uniform is SP and PF} with Theorem~\ref{thm: no need anon Freeman} provides two complementary characterizations. Corollary~\ref{thm: main characterization SP, UFS} says that  the Uniform Phantom mechanism is the only strategyproof mechanism that satisfies UFS; similarly,  the Uniform Phantom mechanism is the only strategyproof mechanism that satisfies PF.
\medskip{}
 
\begin{corollary}[Characterization: UFS/PF and strategyproof]\label{thm: main characterization SP, UFS}
A mechanism satisfies strategyproofness  and UFS (PF)   if and only if it is the Uniform Phantom mechanism. 
\end{corollary}
 \medskip{}
UFS and PF are (strictly) stronger requirements than proportionality, so the characterization given by Corollary~\ref{thm: main characterization SP, UFS} does not hold if UFS or PF are replaced by proportionality. In other words, Theorem~\ref{thm: no need anon Freeman} does not hold if we remove unanimity. A simple example illustrating this can be found in Appendix~\ref{section:prop: cont needed}.

Theorem~\ref{thm: no need anon Freeman} and Corollary~\ref{thm: main characterization SP, UFS} gives the equivalence in Corollary~\ref{corr: ewuiv}. The statements 
are ``tight'': dropping any property in (i), (ii), or (iii) will break the equivalence with (iv).
\medskip{}
\begin{corollary}\label{corr: ewuiv} 
The following are equivalent:
\begin{description}
\item[(i)] $f$ satisfies strategyproofness,  proportionality, and unanimity.
\item[(ii)] $f$ satisfies strategyproofness and UFS.
\item[(iii)] $f$ satisfies strategyproofness and PF.
\item[(iv)] $f$ is the Uniform Phantom mechanism.
\end{description}
\end{corollary}
\medskip{}
A perhaps interesting implication of Corollary~\ref{corr: ewuiv} is that, although combining proportionality and unanimity is a strictly weaker concept than UFS, when combined with strategyproofness the UFS concept is equivalent to requiring both proportionality and unanimity. Similarly, the UFS concept is  strictly weaker concept than PF but, when combined with strategyproofness, PF is equivalent to UFS.

\paragraph{Comparing our results with \citet{FPPV21}.} The Uniform Phantom mechanism appears in~\citet{FPPV21}. \citeauthor{FPPV21}'s Proposition 1 shows that a mechanism is continuous, anonymous, proportional, and strategyproof if and only if it is the Uniform Phantom mechanism. Equivalently, by~\citeauthor{BoJo83}'s (\citeyear{BoJo83})  Corollary 1, \citeauthor{FPPV21}'s characterization holds if continuity is replaced with unanimity.  Our results complement \citeauthor{FPPV21}'s characterization. Firstly, we have shown (in Appendix~\ref{section:prop: cont needed}) that continuity (equivalently, unanimity) is essential for \citeauthor{FPPV21}'s characterization. Secondly, our Theorem~\ref{thm: no need anon Freeman} shows that the anonymity requirement can be removed. 
Studying a slightly different setting, where agents have single-peaked and (possibly) asymmetric preferences, ~\citet{JLPV21} show that neither continuity nor anonymity is required for  \citeauthor{FPPV21}'s characterization. The necessity of unanimity in Theorem~\ref{thm: no need anon Freeman} clarifies a key difference with the setting of symmetric preferences: continuity is required.   Finally,  we provide a more general analysis of fairness axioms in facility location problems and show that the Uniform Phantom mechanism is the unique strategyproof mechanism that satisfies  different combinations of these fairness axioms (Corollary~\ref{corr: ewuiv}).

  
  \section{Equilibria of non-strategyproof,  UFS mechanisms.}\label{section: equilibrium properties}


We now explore the equilibrium properties of non-strategyproof mechanisms. We begin with some terminology.  Given two profiles of locations $\boldsymbol{x}\in[0,L]^n$ and $\boldsymbol{x}'\in[0,L]^n$, we say $\boldsymbol{x}<\boldsymbol{x}'$ if and only if
  $x_i\le x_i'$ for all $i\in N$ and $x_i<x_i'$ for some $i\in N$.  We say a mechanism $f$ is  \emph{strictly monotonic} if 
  $$f(\boldsymbol{x})< f(\boldsymbol{x}') \text{ for all $\boldsymbol{x}<\boldsymbol{x}'$.}$$
  An example of a strictly monotonic mechanism is the ``Average'' mechanism $f_{\text{avg}}(\boldsymbol{x}):=\frac{1}{n}\sum_{i\in N} x_i.$  The Average mechanism is also continuous and satisfies UFS (see Proposition~\ref{prop: average mech satisfies PF} in Appendix~\ref{Section: av mech}). It is clearly not strategyproof. In contrast, the Uniform Phantom mechanism is not strictly monotonic.

Perhaps surprisingly, Theorem~\ref{thm: implementation of UFS} says that the pure Nash equilibrium of any continuous, strictly monotonic, and UFS mechanism has the facility located at the same position as would have been attained by the (strategyproof) Uniform Phantom mechanism. Therefore, in the equilibrium  outcome of such mechanisms, UFS with respect to the agents' true location is satisfied---even if agents misreport their location in equilibrium. This provides an alternative characterization of the Uniform Phantom  mechanism as the equilibrium outcome of any continuous, strictly monotonic, and UFS mechanism.

 To guarantee the existence of a pure Nash equilibrium in Theorem~\ref{thm: implementation of UFS}, we require that each agent's utility function $u_i$ is continuous. Given that, in  our model, each agent's utility function is symmetrically single-peaked, assuming that each $u_i$ is continuous does not affect the set of preferences that are admissible: every symmetrically single-peaked preference on $X$ can be induced by a continuous utility function on $X$.
 
\medskip{}
\begin{theorem}\label{thm: implementation of UFS}
Suppose each agent's utility function $u_i$ is continuous,  and suppose the mechanism $f$ is continuous, strictly monotonic, and satisfies UFS. There exists a pure Nash equilibrium. Furthermore, for every profile of the agents' (true) locations $\boldsymbol{x}$ and every pure Nash equilibrium $\boldsymbol{x}^*$, the equilibrium facility location equals  the facility location of the Uniform Phantom mechanism when agents report truthfully: $f(\boldsymbol{x}^*)=f_{Unif}(\boldsymbol{x})$.
\end{theorem}
\medskip{}

\begin{proof}
The existence of a pure Nash equilibrium follows from~\citep{Debr52,Glic52,Fan52}.  For completeness and following the arguments provided in~\cite{Oz10}, we provide a brief   sketch of the argument.  Naturally, the problem   reduces to the existence of a fixed point solution  to a correspondence $B$ that maps each element of $[0,L]^n$ to a set within $[0,L]^n$. The correspondence $B$ is constructed using each agent's best response correspondence $B_i$, which maps each element of $[0,L]^{n-1}$ to a (non-empty) set within $[0,L]$. Each agent's best response correspondence is well-defined by Weierstrass' Extreme Value theorem---this theorem is applicable because  each agent's utility function $u_i$ is continuous on $[0,L]$.   In this setting, the existence of a fixed point solution is guaranteed by Kakutani's theorem but it requires that $B$ is a convex-valued correspondence and $B$ has a closed graph. The argument for $B$ having a closed graph follows from the standard argument used to prove that every finite game has a mixed strategy Nash equilibrium. The convexity of $B$ follows because each agent's utility function $u_i(f(x_i',\boldsymbol{x}_{-i}'))$ is quasi-concave in their report $x_i'$, which, in turn, follows because $u_i$ is single-peaked and $f$ is continuous and strictly monotonic.

 Now let $\boldsymbol{x}$ be a profile of the agents' (true)  locations, and  let $\boldsymbol{x}^*$ be a pure Nash equilibrium of $f$. Denote by $s_{\text{unif}}:=f_{\text{unif}}(\boldsymbol{x})$  the facility location under the Uniform Phantom mechanism when agents report truthfully. We wish to prove that $f(\boldsymbol{x}^*)=s_{\text{unif}}$. We consider two cases.

\paragraph{Case 1.} Suppose $s_{\text{unif}}=kL/n$ for some $k\in \{0, \ldots, n\}$. By construction of the Uniform Phantom mechanism, it must be that at least $n-k$ agents have true location (weakly) below $s_{\text{unif}}$ and  at least $k$ agents have true location (weakly) above. Now, for the sake of a contradiction, suppose that $f(\boldsymbol{x}^*)<s_{\text{unif}}=kL/n$ (the reverse inequality is treated similarly and therefore is omitted). Notice that there are at least $k$ agents with true location strictly above than $f(\boldsymbol{x}^*)$; let 
$N':=\{ i \in N \ : \ f(\boldsymbol{x}^*)<x_i \}$. If $x_i^*=L$ for all $i\in N'$,  then $f(\boldsymbol{x}^*)\ge kL/n$ (since $f$ satisfies UFS)---a contradiction because $f(\boldsymbol{x}^*)<s_{\text{unif}}=kL/n$. Therefore, $x_i^*<L$ for some agent $i\in N''$. But then $\boldsymbol{x}^*$ cannot be an equilibrium:  agent $i$ can profitably deviate by reporting some $x_i'\in (x_i^*, L]$, which---due to continuity and strict monotonicity of $f$---increases the facility location. 

\paragraph{Case 2.} Suppose $s_{\text{unif}}\in (\frac{kL}{n}, \frac{(k+1)L}{n})$ for some $k\in \{0, \ldots, n-1\}$. By construction of the Uniform Phantom mechanism, it must be that at least $n-k$ agents have true location (weakly) below $s_{\text{unif}}$ and at least $k+1$ agents have true location (weakly) above---note that there are at least $k+1$ agents weakly above $s_{\text{unif}}$ because at least  one agent is located at exactly $s_{\text{unif}}$. Now, for the sake of a contradiction, suppose that $f(\boldsymbol{x}^*)<s_{\text{unif}}$ (the reverse inequality is treated similarly and therefore is omitted). Notice that there are at least $k+1$ agents with location strictly above  $f(\boldsymbol{x}^*)$; let 
$N'':=\{ i \in N \ : \ f(\boldsymbol{x}^*)<x_i \}$. If   $x_i^*=L$  for all $i\in N''$,  then $(k+1)L/n\le f(\boldsymbol{x}^*)$ (since $f$ satisfies UFS)---a contradiction because  $f(\boldsymbol{x}^*)<s_{\text{unif}}\in (\frac{kL}{n}, \frac{(k+1)L}{n})$. Therefore, $x_i^*<L$ for some $i\in N''$. But $\boldsymbol{x}^*$ cannot be an equilibrium: agent $i$ can profitably deviate by reporting  some $x_i'\in (x_i^*, L]$, which---due to continuity and strict monotonicity of $f$---increases the facility location. 
 \end{proof}
\medskip{}
We remark that in a slightly different setting, where agents have single-peaked (and possibly asymmetric) preferences, ~\citet{YaKa13a} provide a general characterization of the equilibrium outcome of  anonymous, continuous, strictly monotonic, and unrestricted-range mechanisms. Although~\citeauthor{YaKa13a}'s results do not formally apply to our setting and do not focus on issues of fairness, our Theorem~\ref{thm: implementation of UFS} is  consistent with their characterization.

An immediate corollary of Theorem~\ref{thm: implementation of UFS} is that the equilibrium outcome of any continuous, strictly monotonic, and UFS mechanism satisfies UFS with respect to the agents' true locations. 
\medskip{}
\begin{corollary}
Suppose each agent's utility function $u_i$ is continuous,   and suppose $f$ is continuous, strictly monotonic, and satisfies UFS.  The output of every (pure) Nash equilibrium of $f$ satisfies UFS with respect to the agents' true location profile. 
\end{corollary}
\medskip{}

Another corollary of Theorem~\ref{thm: implementation of UFS} is that the equilibrium outcome of the average mechanism coincides with the facility location of the Uniform Phantom mechanism when agents report truthfully. In a slightly different setting, where  agents have single-peaked (and possibly asymmetric) preferences,~\citet{ReTr05a} obtain the same result~\citep[see also][]{ReTr11a}. 
\medskip{}
\begin{corollary}
Suppose each agent's utility function $u_i$ is continuous,  and every (pure) Nash equilibrium of the average mechanism coincides with the facility location of the Uniform Phantom mechanism when agents report truthfully.
\end{corollary}
\medskip{}
Unfortunately,  Theorem~\ref{thm: implementation of UFS} cannot be applied to the Nash mechanism's equilibrium outcome since the Nash mechanism (defined in  (\ref{eq: nash mechanism L})) is not strictly monotonic. This can be illustrated via a simple example with 3 agents.
Taking $L=1$, the Nash mechanism maps the location profiles $\boldsymbol{x}=(0, 0.5, 0.9)$ and $\boldsymbol{x}'=(0, 0.5, 1)$ to $0.5$. However,  strict monotonicity requires that $\boldsymbol{x}'$ be mapped to a location strictly higher than $0.5$.


\section{Approximation results}

In this section, we explore the performance of strategyproof and fair mechanisms with respect to two objectives: \emph{total cost minimization} and \emph{welfare maximization}. Rather than make distributional assumptions, we measure the performance of these mechanisms by their worst-case performance over the domain of preference profiles (equivalently, agent locations).

\subsection{Total cost minimization}\label{sec: approx cost} 

A common objective in facility location problems  is to minimize  the total cost of agents: $\sum_{i=1}^n d(y,x_i)$~\citep[see, e.g.,][]{ACLL+20,PrTe13}.  Given a profile of agent locations, $\boldsymbol{x}$ and facility location $y$, we define the
{\em optimal cost} by $\Psi^*(\boldsymbol{x}):=\min_{y\in X} \sum_{i=1}^n d(y,x_i),$ and given a mechanism $f$, let $\Psi_f( \boldsymbol{x})$ denote the total cost attained by the mechanism, i.e., $\Psi_f(\boldsymbol{x}):=\sum_{i=1}^n d( f(\boldsymbol{x}), x_i)$. The mechanism $f$ is a (total cost) $\alpha$-approximation if
\begin{align}\label{Equation: approx 1-cost}
\max_{\boldsymbol{x}\in X^n}\Bigg\{\frac{\Psi_f(\boldsymbol{x})}{\Psi^*(\boldsymbol{x})}\Bigg\}= \alpha.
\end{align}
Notice that $\alpha\ge 1$ for all mechanisms $f$. We refer to a mechanism $f$ with (total cost) $1$-approximation ratio as a \emph{{total cost-optimal} mechanism}.   

We begin by defining the median mechanism, which is known to minimize total cost and, hence, in (\ref{Equation: approx 1-cost}), has a 1-approximation ratio~\citep{PrTe13}.
\medskip{}
\begin{definition}[Median mechanism]\label{def: median mech}
The median mechanism locates the facility at the median of all agents' locations. If there are an even number of agents, the facility is placed at the leftmost of the two middle agent locations.
\end{definition}
 \medskip{}

In addition to being the total cost-optimal mechanism, the median mechanism is strategyproof, anonymous, Pareto efficient, and satisfies unanimity. However, it does not satisfy our weakest notions of proportionality-based fairness: IFS or proportionality. 

Proposition~\ref{prop: IFS bound cost} provides a stark negative result. Any mechanism that is strategyproof, anonymous, unanimous and satisfies IFS has total cost approximation of exactly $n-1$, which is unbounded as $n$ grows large.  

 \medskip{}

\begin{proposition}\label{prop: IFS bound cost} 
 Any strategyproof, anonymous, unanimous mechanism that satisfying IFS has a total cost approximation of  $n-1$. As $n\rightarrow \infty$, this approximation is unbounded.
\end{proposition}
\medskip{}

Proposition~\ref{prop: IFS bound cost} implies that, on the basis of total cost approximation, there is no difference between any of the mechanism characterized in Sections~\ref{section:char1} and~\ref{section:char2}. This suggests the need for an alternative (or additional ``tie-breaking'') performance measure that is more sensitive to proportionality-based fairness axioms. In the next subsection, we adopt an alternative performance that appears in the literature and allows for a more nuanced analysis.  

\subsection{Welfare maximization}\label{sec: approx} 

Within but also beyond facility location problems, a common objective in collective decision-making is to maximize (utilitarian or social) welfare.  Given a profile of locations $\boldsymbol{x}$ and a facility location $y$, the (utilitarian or social) \emph{welfare}    is   defined as the sum of the utilities of the agents: $\sum_{i=1}^n u_i(y)$. In our setting, agents' utility functions are unknown by the mechanism designer---in fact, it is impossible for the mechanism designer to elicit more information about agents' utilities than their peak location without violating strategyproofness~\citep[see, e.g.,][]{BaJa94,Weym08}. Therefore, it is necessary to assume a specific functional form as a proxy of agents' utilities (alternatively, one may simply assume that agents' utilities are all of a specific functional form). Importantly,  this functional form must be non-negative to have a well-defined welfare-maximization approximation problem (as will be described by~(\ref{Equation: approx 1})); this requirement rules out the functional form $-d(y, x_i)$ that appeared in Section~\ref{sec: approx cost}. Notice that the total cost minimization problem (\ref{Equation: approx 1-cost}) can equivalently be described as minimizing the total social 
\emph{disutility},  $\sum_{i=1}^n -u_i(y)$, when each agent is assumed to have utility function $u_i(y)=-d(y, x_i)$.

 We focus on the following functional form:   
\begin{align}\label{eq: welfare ob}
\sum_{i=1}^n u_i(y):=\sum_{i=1}^n \big(L-d(y,x_i)\big),
\end{align}
which  appears in other facility location papers~\cite[see, e.g.,][]{AnDe18,ACLP20,DeFRV23,ZoMi15}. One could consider alternative linear utility functions, such as $u_i(y)=L' -d(y,x_i)$ with $L'>L$, this will always lead to a welfare approximation ratio (to be defined in (\ref{Equation: approx 1})) that is strictly less than that obtained with $L'=L$. Therefore, our choice of $L'=L$ is the most conservative among this family of linear utility functions. 

 
We explore the performance of strategyproof and fair mechanisms with respect to \emph{welfare maximization} of (\ref{eq: welfare ob}).    Given a profile of agent locations, $\boldsymbol{x}$ and facility location $y$, we define the
{\em optimal welfare} by $\Phi^*(\boldsymbol{x}):=\max_{y\in X} \sum_{i=1}^n (L-d(y,x_i)),$ and given a mechanism $f$, let $\Phi_f( \boldsymbol{x})$ denote the welfare attained by the mechanism, i.e., $\Phi_f(\boldsymbol{x}):=\sum_{i=1}^n (L-d(f(\boldsymbol{x}),x_i))$. The mechanism $f$ is a (welfare) $\alpha$-approximation if
\begin{align}\label{Equation: approx 1}
\max_{\boldsymbol{x}\in X^n}\Bigg\{\frac{\Phi^*(\boldsymbol{x})}{\Phi_f(\boldsymbol{x})}\Bigg\}= \alpha.
\end{align}
Notice that $\alpha\ge 1$ for all mechanisms $f$. We refer to a mechanism $f$ with (welfare) $1$-approximation ratio as a \emph{{welfare-optimal} mechanism}.    Before proceeding to our analysis, we discuss briefly the distinction between the total cost minimization and welfare maximization approximation problems.


\paragraph{Total cost minimization vs welfare maximization.} Minimizing the total cost (Section~\ref{sec: approx cost})  and maximizing welfare, as in (\ref{eq: welfare ob}), are equivalent optimization problems. Indeed, the total cost objective function is a simple translation of the welfare objective function. Therefore, both problems have the same ``optimal'' mechanism: the median mechanism (Definition~\ref{def: median mech}), which is strategyproof, anonymous, Pareto efficient, and  unanimous  but   does not satisfy  IFS or proportionality.   However, in general, when considering approximately-optimal mechanisms, the welfare approximation ratio  of a mechanism (\ref{Equation: approx 1}) will not  equal  the total cost approximation ratio (\ref{Equation: approx 1-cost}). Indeed, the total cost approximation analysis in Section~\ref{sec: approx cost} led to a stark negative result. As will be shown, focusing on  our welfare maximization objective (\ref{eq: welfare ob}) allows for a more nuanced evaluation of the performance of various mechanisms and  a clearer analysis of the tradeoffs imposed by our proportionality-based fairness axioms for welfare maximization.  

The key distinction between the total cost approximation and welfare approximation can be intuitively understood by considering instances that might generate a large approximation ratio. In the welfare formulation,  the denominator in the  ratio (\ref{Equation: approx 1}) is  the total welfare generated by the mechanism $f$. This denominator is small if the mechanism locates the facility far away from many agents. In the case of the optimal median mechanism, welfare is minimized when half of the agents are located at each extreme location. In contrast, in the total cost formulation, the denominator  (\ref{Equation: approx 1-cost}) is the total cost generated by the optimal (median) mechanism. This denominator is zero or close to zero if all agents are closely located.  Therefore, the total cost approximation analysis places greater weight on instances where the optimal median mechanism may achieve a perfect or near-perfect solution with total cost approximately zero. Whereas the welfare approximation analysis may be viewed as more egalitarian: it places greater weight on instances where a mechanism  generates very little welfare, perhaps because many agents are located at opposite extremes. A priori both approximation approaches appear useful and neither appears more desirable than the other. However, given the stark  total cost approximation results in Section~\ref{sec: approx cost} that fails to differentiate between various families of strategyproof and proportionally fair mechanisms, the welfare approximation approach is a useful additional performance measure---even if only used as a tie-breaking rule. \\


We now proceed to our analysis. Lemma~\ref{lemma: IFS bound} provides a welfare approximation lower bound for mechanisms that satisfy IFS.
 \medskip{}

\begin{lemma}\label{lemma: IFS bound} 
 Any mechanism satisfying IFS has a welfare approximation of at least $1+\frac{n-2}{n^2-2n+2}$. As $n\rightarrow \infty$, this lower bound approaches $1$.
\end{lemma}
\medskip{}

We now provide an example of an IFS mechanism, which we call the Constrained Median mechanism,  that obtains the welfare approximation of Lemma~\ref{lemma: IFS bound}. 
The Constrained Median mechanism locates the facility at the median location whenever the median  location lies in the interval $[L/n, L(1-1/n)]$. When the median location is below $L/n$ (resp., above $L(1-1/n)$), the facility is located at the minimum of $L/n$ and maximum-agent report (resp., maximum of $L(1-1/n)$ and the minimum-agent report).  Definition~\ref{def: constrained med} provides a formal definition, and Figure~\ref{figure:rules} provides an illustration of the mechanism.

 \medskip{}

\begin{definition}[Constrained Median]\label{def: constrained med}
The Constrained Median mechanism $f_{\text{CM}}$ is a phantom mechanism that places $\ceil{\frac{n-1}{2}}$ phantoms at $L/n$ and the remaining phantoms at $L(1-\frac{1}{n})$. 
\end{definition}
		
\medskip{}

Theorem~\ref{thm: IFS bound} says that the Constrained Median mechanism obtains the best welfare approximation guarantee among all IFS mechanisms, including non-strategyproof mechanisms. Furthermore, the Constrained Median mechanism can easily be seen to not only satisfy IFS but also to be strategyproof, anonymous, and unanimous (Theorem~\ref{char: IFS and SP}).
 \medskip{}
\begin{theorem}\label{thm: IFS bound}
Among all IFS mechanisms, the Constrained Median  mechanism provides the best welfare approximation guarantee, i.e.,  it achieves the approximation ratio in Lemma~\ref{lemma: IFS bound}.
\end{theorem}
\medskip{}
The intuition behind the  welfare  approximation ratio converging to $1$ is that as $n$ approaches infinity, the phantoms placed at $L/n$ (and $L(1-1/n)$) converge to $0$ (and $L$), and hence the Constrained mechanism mechanism converges to the median mechanism.

Lemma~\ref{lemma: UFS bound} provides a minimum  welfare approximation bound for mechanisms that satisfy UFS (or  proportionality or PF).
\medskip{}
\begin{lemma}\label{lemma: UFS bound}
Any mechanism satisfying UFS (or proportionality or PF) has a welfare approximation of at least 
\begin{align}\label{eq: lemma UFS lower bound}
\max_{k\in \mathbb{N} \ :\  0\le k\le n/2 } \frac{n(n-k)}{k^2+(n-k)^2}.
\end{align}
As $n\rightarrow \infty$, this lower bound approaches $\frac{\sqrt{2}+1}{2}\approx 1.207$.
\end{lemma}
\medskip{}

We now show that the Uniform Phantom mechanism   obtains the welfare approximation of Lemma~\ref{lemma: UFS bound}. This means that the Uniform Phantom mechanism provides the best welfare approximation guarantee among all UFS (or proportional or PF) mechanisms, including non-strategyproof mechanisms. Furthermore, from Theorem~\ref{thm: no need anon Freeman}, we know that the Uniform Phantom mechanism has the added benefit of being strategyproof, anonymous,  and unanimous.

\medskip{}
\begin{theorem}\label{thm: UFS bound}
Among all UFS (or proportional or PF) mechanisms, the Uniform Phantom mechanism provides the best welfare approximation guarantee, i.e.,  it achieves the approximation ratio in Lemma~\ref{lemma: UFS bound}.
\end{theorem}
\medskip{}

Figure~\ref{figure: approx} illustrates the approximation results of this section.

\begin{figure}[!htb]
        \centering
        \includegraphics[height=0.4\textheight]{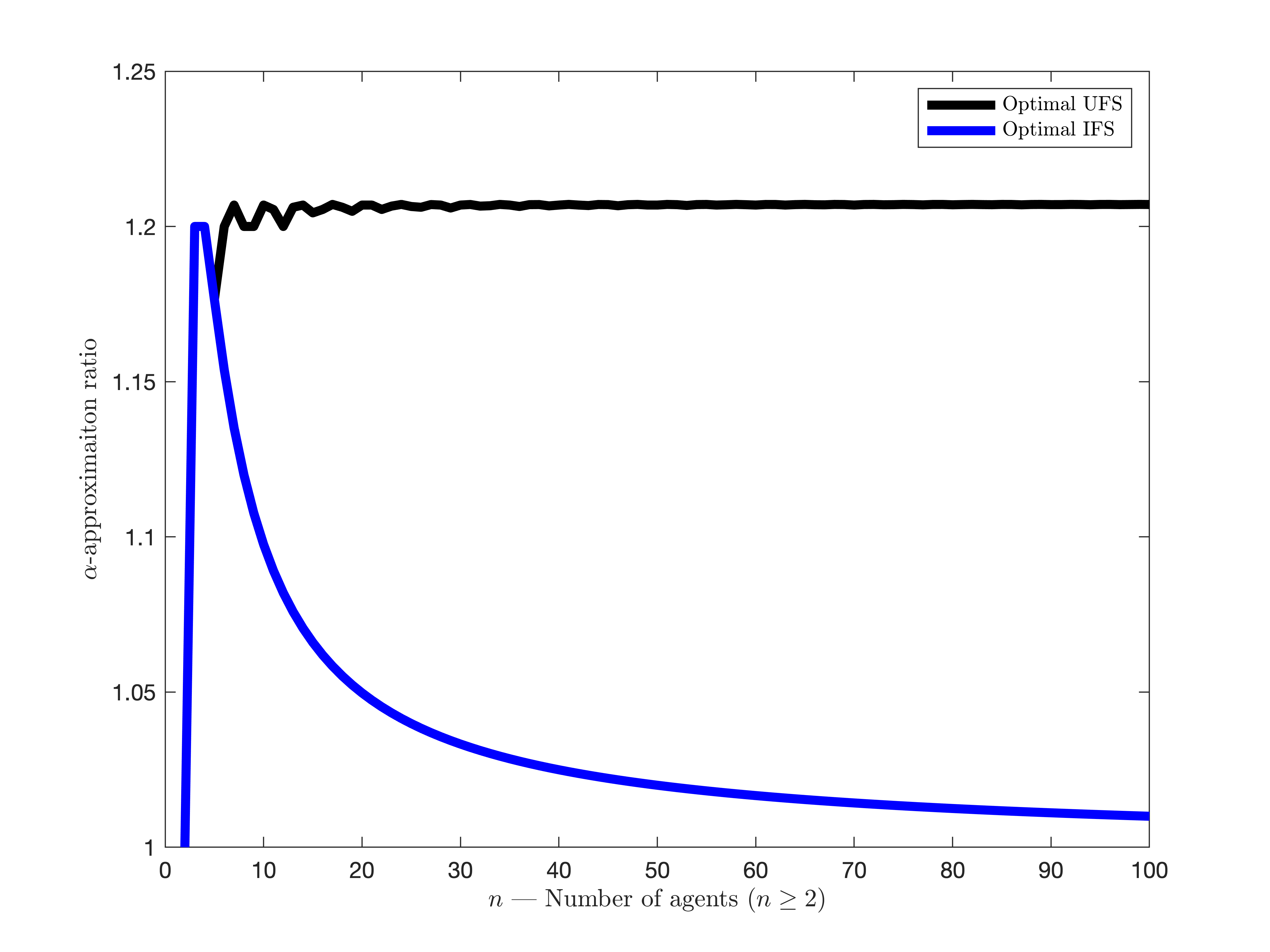}
        \caption{ {\footnotesize{The best welfare approximation guarantee for mechanisms that satisfy UFS and IFS. }}}
\label{figure: approx} 
\end{figure}


\section{Discussion and directions for future research.}\label{sec: discussion}


 Facility location is a classical problem in economic design. In this paper, we provided a deeper understanding of strategyproof and proportionally fair  
 mechanisms.   Table~\ref{table:summary} provides an overview of most of the mechanisms considered in the paper and the properties they satisfy. Our results provide strong support for the desirability of the Uniform Phantom mechanism in terms of satisfying fairness and strategyproofness.

 \begin{table*}[t!]
 			    			\centering
 			    			\caption{Summary of results. All mechanisms are also unanimous, anonymous and Pareto efficient. Proofs of the results for the Average mechanism can be found in Appendix~\ref{Section: av mech}. The welfare approximation results for the Nash and Midpoint mechanisms are from~\cite{LAW21}, and the total cost approximation results for those mechanisms can be found in Appendix~\ref{appsec: totalcostapprox}.}
 			    			\label{table:summary}
 			    			\scalebox{0.8}{
 			    			\begin{tabular}{clclclclclclc|clc}
 			    			\hline
 			    				\toprule
 			    				\textbf{Mechanism}& \textbf{Strategyproof} & \textbf{PF} & \textbf{UFS} & \textbf{Proportionality} & \textbf{IFS} & \textbf{Util-approx (limit)} & \textbf{Cost-approx} \\
 			    				\midrule
 			    				\hline
 								 \textbf{Uniform Phantom }& \underline{Yes} & \underline{Yes}&\underline{Yes} &\underline{Yes} &\underline{Yes} & $\frac{\sqrt{2}+1}{2}\approx 1.207$ & $n-1$ \\
 								\textbf{Median}& \underline{Yes} & {\color{gray}No} &{\color{gray}No}  &{\color{gray}No}  &{\color{gray}No}  & $1$ & $1$\\						    
 								\textbf{Constrained Median}& \underline{Yes} & {\color{gray}No} &{\color{gray}No}  &{\color{gray}No} &\underline{Yes} & $ 1$ & $n-1$\\
 			    \textbf{Nash mechanism}& {\color{gray}No} & \underline{Yes}&\underline{Yes} &\underline{Yes} &\underline{Yes} & $\in [\frac{\sqrt{2}+1}{2}, 2]$ & $\in [2-\frac{2}{n},\frac{n}{2}]$\\
 			    			    \textbf{Midpoint mechanism}& {\color{gray}No}  & {\color{gray}No} &{\color{gray}No}  &{\color{gray}No}  &\underline{Yes} & $2$ & $\frac{n}{2}$\\
 \textbf{Average mechanism}& {\color{gray}No} & \underline{Yes}&\underline{Yes} &\underline{Yes} &\underline{Yes} & $\frac{\sqrt{2}+1}{2}$ & $2-\frac{2}{n}$\\

\hline
 
 			    				\bottomrule
 			    			\end{tabular}
 			    			}

 			    		\end{table*}

 	Moving beyond the fairness axioms that we presented,  one can also consider stronger notions of proportionality-based fairness. For example,  the following property, which we call  Strong Proportional Fairness (SPF), is stronger than PF.  Given a profile of locations $\boldsymbol{x}$ within range of distance $R$, a facility location $y$ satisfies \emph{Strong Proportional Fairness (SPF)} if, for any subset of voters $S\subseteq N$ within a range of distance $r$,  the location should be at most $R\frac{n-|S|}{n}+r$ distance from each  agent in $S$, i.e., $d(y, x_i)\le R\frac{n-|S|}{n}+r$ for all $i\in S$. 
	
	However, it can be easily shown that the Uniform Phantom mechanism does not satisfy SPF. Our result (that the Uniform Phantom mechanism is the only SP and PF mechanism) then implies that there exists no strategyproof and SPF mechanism. In this sense, the compatibility between strategyproofness and fairness axioms ceases to hold when we move from PF to SPF. 
	 
There are several directions for future work to build on the framework and results that we have presented. For example, it may be fruitful to extend our analysis  to incorporate a facility with capacity constraints, multiple facilities,  alternative fairness concepts, considering weaker notions of strategyproofness, or alternative utility functions that are not necessarily single-peaked.  Considering alternative utility functions can implicitly allow for behavioral assumptions in how agents use or benefit from the facility. For example, if agents do not benefit at all from the facility location when it is beyond a threshold distance from their ideal location, then this would correspond to a utility function that it is single-peaked but, beyond a certain threshold distance from their peak location, the utility function becomes constant and takes its minimal value~\citep[see, e.g.,][]{ZZML23}.    An important direction is also to extend the  strategyproof and proportionally fair facility location problem to multiple dimensions. Although some real-world problems (such as the provision of public goods) are well-suited to a unidimensional setting, other real-world problems are better suited to a multidimensional setting. By leveraging existing strategyproofness results for the  multidimensional facility location problem (such as Theorem~1 in~\cite{BoJo83}, which applies to settings where agents have separable preferences) and developing appropriate multidimensional generalizations of the proportionality-based fairness axioms that we presented, we believe progress can be made on this.

\newpage
{\small
	\setstretch{.75}
\bibliographystyle{bib/aer2}
		  \bibliography{bib/abb.bib,bib/ref_new.bib}
}

\appendix

\section{Omitted proofs.}

\subsection{Proof of Proposition~\ref{prop: UFS implies unanimity}.}\label{section: prop: UFS implies unanimity}

\begin{proof}
Point (i): {We wish to prove that UFS implies proportionality, IFS, and unanimity.} Let $\boldsymbol{x}$ be an arbitrary location profile and let $y$ be a facility location that satisfies UFS. From the definition of UFS, it is immediate that IFS and unanimity are satisfied. It remains to prove that proportionality is satisfied. For the sake of a contradiction, suppose that proportionality is not satisfied. That is, $\boldsymbol{x}$ is such that $x_i\in \{0,L\}$ for all $i\in N$ and  $y\neq L\frac{|i\in N \ : \ x_i=1|}{n}$. Let $k=|i\in N \ : \ x_i=L|$. If $k=0$, then UFS requires that $y=0$, and proportionality is satisfied---a contradiction. If $k>0$, then UFS requires that:
$$|L-y|\le L(1-\frac{k}{n}), \quad \text{i.e.,}\quad y\ge \frac{kL}{n}  \qquad \text{and} \qquad  |0-y| \le L(1- \frac{n-k}{n}), \quad \text{i.e.,}\quad y\le \frac{kL}{n}.$$
The inequalities above imply that $y=\frac{kL}{n}$ and proportionality is satisfied---a contradiction.

Point (ii):  {We wish to prove that PF implies UFS.} This follows immediately by noting that a set of agents all located at the same location are within a range of distance $r=0$. Taking $r=0$ in the PF definition shows that PF implies UFS. 

It is straightforward to see that the relations in the proposition are strict and also that there is no logical relation between proportionality, IFS, and unanimity. We omit the proofs.
 \end{proof}

 \subsection{Proof of Proposition~\ref{prop: Egal Nash}.}\label{section: prop: Egal Nash}

\begin{proof}
Point (i): {We wish to prove that the median mechanism satisfies unanimity and strategyproofness, but does not satisfy IFS, PF, UFS, nor Proportionality.} The median mechanism is known to be strategyproof \citep{PrTe13}; it is also clearly unanimous. Finally, consider the agent location profile with $n-1$ agents at $0$ and $1$ agent at $L$. The median mechanism locates the facility at $0$, which violates both IFS and Proportionality (and hence also UFS and PF).

Point (ii): {We wish to prove that the midpoint mechanism satisfies IFS and unanimity, but does not satisfy  strategyproofness, PF, UFS, nor Proportionality.} The midpoint mechanism places the facility at the average of the leftmost and rightmost agent. It is therefore unanimous but not strategyproof. The maximum cost that can be incurred by an agent is $L/2$, which is obtained when the leftmost agent is at $0$ and the rightmost agent is at $L$. However, this IFS is satisfied since $n\ge 2$. To see that the midpoint mechanism does not satisfy Proportionality, consider the agent location profile with 2 agents at $0$ and 1 agent at $L$. The midpoint mechanism places the facility at $L/2$, but Proportionality requires that the mechanism is placed at $L/3$. Since Proportionality is not satisfied, UFS and PF are also not satisfied.

Point (iii): {We wish to prove that the Nash mechanism satisfies PF but is not strategyproof. To this end, we first define a notion of monotonicity that requires that if a location profile is modified by an agent shifting its location, the facility placement under the modified profile will not  shift in the opposite direction. Definition~\ref{def: Mono} formalizes this notion.}
\medskip{}
\begin{definition}[Monotonic]\label{def: Mono}
A mechanism $f$ is \emph{monotonic} if
\[f(\boldsymbol{x})\leq f(\boldsymbol{x}')\]
for all $f(\boldsymbol{x})$ and $f(\boldsymbol{x}')$ such that $x_i \leq x_i'$ for all $i\in N$ and $x_i<x_i'$ for some $i\in N$.
\end{definition}
\medskip{}
We next prove the following auxiliary lemma.

\begin{lemma}\label{lem: ufs mono is PF}
A mechanism that satisfies UFS and monotonicity also satisfies PF.
\end{lemma}

\begin{proof}
Let $\boldsymbol{x}$ be an arbitrary agent location profile, and $f$ be a mechanism that satisfies UFS and monotonicity. Consider the set $S=\{1,\dots,m\}\subset N$ of $m$ agents and denote $r:= \max_{i\in S}\{x_i\}-\min_{i\in S}\{x_i\}$. We prove that the maximum distance of the facility from any agent in $S$ is at most $L\frac{n-m}{n}+r$.

Denote $f:=\argmax_{i\in S}\{d(f(\boldsymbol{x}),x_i)\}$ as the agent in $S$ whose location $x_f$ under $\boldsymbol{x}$ is furthest from the respective facility location, which is either $\max_{i\in S}\{x_i\}$ or $\min_{i\in S}\{x_i\}$. Consider the modified profile $\boldsymbol{x}'$ where $x_i'=\max_{i\in S}\{x_i\}$ for all $i\in S$ if $f(\boldsymbol{x})\geq x_f$ and $x_i'=\min_{i\in S}\{x_i\}$ for all $i\in S$ if $f(\boldsymbol{x})< x_f$. Also, $x_i'=x_i$ for all $i\in N\backslash S$. In other words, the agents in $S$ have their locations moved to the rightmost agent in $S$ if the facility is weakly right of the furthest agent of $S$ under $\boldsymbol{x}$. If the facility is strictly left of the furthest agent of $S$ under $\boldsymbol{x}$, the agents in $S$ have their locations moved to the leftmost agent.

Due to monotonicity, the facility does not move closer to $x_f$ when modifying $\boldsymbol{x}$ to $\boldsymbol{x}'$, so we have
\[d(f(\boldsymbol{x}'),x_f)\geq d(f(\boldsymbol{x}),x_f).\]
Note that all $m$ agents of $S$ are at the same location under $\boldsymbol{x}'$. Denote this location as $x_S'$. Due to UFS, we also have
\[d(f(\boldsymbol{x}'),x_S')\leq L\frac{n-m}{n}.\]
We therefore have
\begin{align*}
d(f(\boldsymbol{x}),x_f)&\leq d(f(\boldsymbol{x}'),x_f)\leq d(f(\boldsymbol{x}'),x_S')+r\leq L\frac{n-m}{n}+r.
\end{align*}

\end{proof}
\medskip{}
The Nash mechanism is known to satisfy UFS and monotonicity \citep{LAW21}. It therefore satisfies PF. 
\end{proof}
\medskip{}

  \subsection{Tightness of Theorem~\ref{char: IFS and SP}.}\label{section:prop: ess char}
%

\begin{proof}
We wish to prove that each of the requirements in Theorem \ref{char: IFS and SP} are necessary for the theorem to hold. We  show that if any one of the requirements (i.e., strategyproofness, unanimity, anonymity, and IFS) are removed, then   Theorem~\ref{char: IFS and SP}---not only fails to hold---but there exists such a mechanism. We do this by providing examples  of mechanisms that fulfill all but one of the requirements of Theorem~\ref{char: IFS and SP} but that are  not  phantom mechanisms with the $n-1$ phantoms contained in $[L/n, L(1-1/n)]$.

\textbf{Strategyproofness.} By Proposition~\ref{prop: Egal Nash}, the midpoint mechanism is an example of a mechanism that is not strategyproof, but satisfies unanimity, anonymity and IFS. {However, the midpoint mechanism is not a phantom mechanism---this follows immediately because phantom mechanisms are necessarily strategyproof.}

\textbf{Unanimity.} For $n\ge 2$, consider the constant-$\frac{L}{2}$ mechanism, whereby the facility is always located at $\frac{L}{2}$. This mechanism is clearly strategyproof and anonymous and does not satisfy unanimity. Furthermore, it satisfies IFS because the largest cost that any agent can experience is $\frac{L}{2}$,  which is (weakly) lower than $L(1-\frac{1}{n})$ for any $n\ge 2$. However, the constant-$\frac{L}{2}$ mechanism is not a phantom mechanism---this follows immediately because phantom mechanisms necessarily satisfy unanimity.


\textbf{Anonymity.} For simplicity take $n=3$ and consider the mechanism $f$ that locates the facility at 
 \begin{align*}
f(\boldsymbol{x}):=& \max\{\min\{x_1,\frac{L}{3}+\varepsilon\},\min\{x_2,\frac{L}{3}\},\min\{x_3,\frac{L}{3}\},\\
 &\qquad \min\{x_1,x_2,L(1-\frac{1}{3}) \},\min\{x_2,x_3,L(1-\frac{1}{3}) \}, \min\{x_1,x_3,L(1-\frac{1}{3}) \},\\
 &\qquad \min\{x_1,x_2,x_3,L\},0\},
 \end{align*}
  where $\varepsilon>0$ is sufficiently small. This is a Generalized Median mechanism~\citep{BoJo83} and, hence, is strategyproof. Furthermore, it is easy to see that $f$ is unanimous. \citeauthor{BoJo83}'s (\citeyear{BoJo83}) Lemma~3 then says that $f$ is Pareto efficient.  However, the mechanism is not anonymous: for $\boldsymbol{x}=(0,0,L)$,  $f(\boldsymbol{x})=L/3$, but for $\boldsymbol{x}'=(L,0,0)$,  $f(\boldsymbol{x}')=L/3+\varepsilon$. 
  
  We now show that the mechanism satisfies IFS. Since the mechanism is Pareto efficient, IFS is trivially satisfied if $x_i\leq L(1-\frac{1}{3})$ for all $i\in N$ or if $x_i \geq \frac{L}{3}$ for all $i\in N$. Now consider   some location profile $\boldsymbol{x}$ that does not belong to these trivial cases, i.e., there is at least one agent with location below $\frac{L}{3}$ (resp., $L(1-\frac{1}{3})$)  and at least one agent with location above $\frac{L}{3}$ (resp., $L(1-\frac{1}{3})$). In these cases, IFS can only possibly be violated if $f(\boldsymbol{x})>L(1-\frac{1}{3})$ or $f(\boldsymbol{x})<\frac{L}{3}$.    However, $f(\boldsymbol{x})>L(1-\frac{1}{3})$ if and only if $x_i>L(1-\frac{1}{3})$ for all $i\in N$---but the latter condition does not hold.  Similarly, $f(\boldsymbol{x})<\frac{L}{3}$ if and only if $x_i<\frac{L}{3}$ for all $i\in N$---but, again, the latter condition does not hold. Therefore, we conclude that IFS is satisfied. {However, the mechanism $f$ is not a phantom mechanism---this follows immediately because phantom mechanisms are necessarily anonymous.}

\paragraph{IFS.} The Phantom mechanism that places all $n-1$ phantoms at $0$ is strategyproof, unanimous and anonymous. However, it does not satisfy IFS as the facility can be placed at $0$ when there is an agent at $L$. {It is immediate that this Phantom mechanism violates the condition of the theorem that all phantoms are located in the interval $[L/n, L(1-1/n)]$.}
\end{proof}

\medskip{} 

\subsection{Proof of Proposition~\ref{prop: uniform is SP and PF}.}\label{section:prop: uniform is SP and PF}
\begin{proof}
The Uniform Phantom mechanism is strategyproof since it is a Phantom mechanism and all Phantom mechanisms are strategyproof~\citep[see, e.g., Corollary 2 of][]{MaMo11a}. We now prove that the Uniform Phantom mechanism satisfies PF. Let $\boldsymbol{x}$ be an arbitrary location profile and let $S=\{1,\dots,s\}\subseteq N$ be a set of $s$ agents; denote $r:= \max_{i\in S}\{x_i\}-\min_{i\in S}\{x_i\}$. We prove  $d(f(\boldsymbol{x}), x_i)\le L\frac{n-s}{n}+r$ for all $i\in S$. If $r=L$, then the result is trivially true. Suppose that $r<L$. If the location is within the range of the agents in $S$, PF is immediately satisfied. Next we consider the case there the location is outside the range of the agents in $S$. Recall that the Uniform Phantom mechanism places the facility at the $n$-th entity of the $2n-1$ phantoms and agents. 
There are at least $s$ agents in the range of the locations of the agents in $S$, so the facility is at most $n-s$ phantoms away from the nearest agent in $S$. Since the distance between adjacent phantoms is $L/n$,  the facility is at most distance $(n-s)L/n$ from the nearest agent in $S$. Hence, the maximum distance of the facility from any agent in $S$ is $L\frac{n-s}{n}+r$.
\end{proof}
\medskip{}

  \subsection{Lemma~\ref{lemma: SP, unan, prop implies anon} and proof of Lemma~\ref{lemma: SP, unan, prop implies anon}.}\label{section:lemma: SP, unan, prop implies anon}

	\begin{lemma}\label{lemma: SP, unan, prop implies anon}
	A mechanism that is strategyproof, unanimous, and proportional must also be anonymous.
	\end{lemma}
	
	\begin{proof}
	Suppose $f$ is strategyproof and satisfies proportionality and unanimity. We wish to show that $f$ is anonymous (Definition~\ref{def: anon}). First we note that by~\citeauthor{BoJo83}'s (\citeyear{BoJo83}) Proposition 2, any unanimous and strategyproof mechanism must satisfy the following uncompromising property.  
	
	\medskip{}
    \begin{definition}[Uncompromising]\label{Def: uncomp}
    A mechanism $f$ is \emph{uncompromising} if, for every profile of locations $\boldsymbol{x}$, and each agent $i\in N$,  if $f(\boldsymbol{x})=y$ then
\begin{align}\label{eqaution: uncompromising 1}
x_i>y &\implies f(x_i',\boldsymbol{x}_{-i})=y\qquad \text{ for all } x_i'\ge y &&\text{ and,}\\
x_i<y &\implies f(x_i' , \boldsymbol{x}_{-i})=y \qquad \text{ for all } x_i'\le y.\label{eqaution: uncompromising 2}
\end{align}
\end{definition}
	\medskip{}

Now consider an arbitrary profile of locations $\boldsymbol{x}$ and an arbitrary permutation of the profile $\boldsymbol{x}$, which we denote by $\boldsymbol{x}_\sigma$. We will show that $f(\boldsymbol{x})=f(\boldsymbol{x}_\sigma)$. First note that if $\boldsymbol{x}$ is such that $x_i=c$ for some $c\in [0,L]$, then $f(\boldsymbol{x})=f(\boldsymbol{x}_\sigma)$ by unanimity. Therefore, we assume that $\boldsymbol{x}$ is such that $x_i\neq x_j$ for some $i,j\in N$. 

\paragraph{Case 1.} Suppose that $f(\boldsymbol{x})\neq x_i$ for any $i\in N$. Recall that \citeauthor{BoJo83}'s (\citeyear{BoJo83}) Lemma~3 says that any strategyproof and unanimous mechanism is Pareto efficient; therefore, $ \min_{i\in N} x_i \le f(\boldsymbol{x})\le  \max_{i\in N} x_i$. Now if all agents strictly below (resp., above) $f(\boldsymbol{x})$ shift their location to $0$ (resp., $L$), then, by the uncompromising property, the facility location must be unchanged. Let $\boldsymbol{x}'$ denote this augmented location profile and let $k'$ denote the number of agents with $x_i'=L$. By proportionality, it must be that $f(\boldsymbol{x})=f(\boldsymbol{x}')=\frac{k'L}{n}$. Now consider the permutation of the profile $\boldsymbol{x}'$, i.e., $\boldsymbol{x}_\sigma'$. The implication of the proportionality property is independent of agent labels; therefore, $f(\boldsymbol{x}_\sigma')=f(\boldsymbol{x}')$.  Now   shift the agent locations in $\boldsymbol{x}_\sigma'$ so that they replicate the permuted location profile $\boldsymbol{x}_\sigma$---note that this process only involves agents strictly above (resp., below) $f(\boldsymbol{x}_\sigma')$ moving to a location above (resp., below) $f(\boldsymbol{x}_\sigma')$. Therefore, by the uncompromising property, it must be that $f(\boldsymbol{x}_\sigma')=f(\boldsymbol{x}_\sigma)$. 
Combining the three sets of equalities gives  
$$f(\boldsymbol{x})=f(\boldsymbol{x}')=f(\boldsymbol{x}_\sigma')=f(\boldsymbol{x}_\sigma).$$
That is, the facility location is unchanged by permutations, i.e.,  anonymity  is satisfied.

\paragraph{Case 2.}  Suppose that $f(\boldsymbol{x})= x_i$ for some $i\in N$.  Let $M\subseteq N$ be the subset of agents with $x_i=f(\boldsymbol{x})$.    Let $M_0, M_1\subseteq N$ correspond to the subset of agents with location strictly below and strictly above $f(\boldsymbol{x})$, respectively. Denote $|M_0|=k_0$ and $|M_1|=k_1$. We first show that 
\begin{align}\label{eq 1}
\frac{k_1L}{n}\le f(\boldsymbol{x}).
\end{align}
 For the sake of  contradiction, suppose that (\ref{eq 1}) does not hold (i.e., $f(\boldsymbol{x})<\frac{k_1L}{n}$), and consider  the location profile $\boldsymbol{x}'$ obtained by modifying $\boldsymbol{x}$ such that the $M_0$ (resp., $M_1$) agents' locations are shifted to  $0$ (resp., $L$) and the other agents' (i.e., those in $M$) have location unchanged. By the uncompromising property, $f(\boldsymbol{x}')=f(\boldsymbol{x})$. Now consider the modified location profile $\boldsymbol{x}''$ such that $x_i''=x_i'$ for all $i\notin M$ and $x_i''=0$ for all $i\in M$.  By proportionality,  $f(\boldsymbol{x}'')=\frac{k_1L}{n}$ and, by supposition that $f(\boldsymbol{x})<\frac{k_1L}{n}$,  we have
\begin{align}\label{eq1 contra}
f(\boldsymbol{x}')=f(\boldsymbol{x})<\frac{k_1L}{n}=f(\boldsymbol{x}'').
\end{align}
Now notice that the  profile $\boldsymbol{x}'$ can be obtained from $\boldsymbol{x}''$ by shifting the subset of $M$ agents' locations from $0$ to $f(\boldsymbol{x})$, which is to the left of $f(\boldsymbol{x}'')$. The uncompromising property then requires that 
$$f(\boldsymbol{x}')=f(\boldsymbol{x}'')=\frac{k_1L}{n},$$
which contradicts~(\ref{eq1 contra}). We conclude that~(\ref{eq 1}) holds.

With condition~(\ref{eq 1}) in hand, we can now proceed by considering two subcases.

\paragraph{Subcase 2a.}  Suppose that $f(\boldsymbol{x})=\frac{kL}{n}$ for some $k\in \{0, \ldots, n\}$. Consider any  profile of locations $\boldsymbol{x}'\in \{0,L\}^n$ with $k$ agents at location $L$. By proportionality, $f(\boldsymbol{x}')=\frac{kL}{n}=f(\boldsymbol{x})$. Note that the proportionality axiom is independent of agent labels and, by~(\ref{eq 1}),  $f(\boldsymbol{x})\ge \frac{k_1L}{n} \implies k\ge k_1$. Therefore,  the permuted location of profile $\boldsymbol{x}_{\sigma}$ can be attained by relabeling agents in $\boldsymbol{x}'$ and then shifting their reports from $L$ (resp., $0$) to their original location that is weakly above (resp., below) $f(\boldsymbol{x}')=\frac{kL}{n}$---by the uncompromising property, the facility location will not change.  Therefore, $f(\boldsymbol{x}_\sigma)=f(\boldsymbol{x}')=\frac{kL}{n}=f(\boldsymbol{x})$, as required.

\paragraph{Subcase 2b.} Suppose that $f(\boldsymbol{x})\neq \frac{kL}{n}$ for any $k\in \{0, \ldots, n\}$. Let $k^*$ be the smallest integer such that $f(\boldsymbol{x})< \frac{k^*L}{n}$; by~(\ref{eq 1}), $k^*>k_1$.   Now consider any location profile $\boldsymbol{x}'\in \{0,L\}^n$ with exactly $k^*$ agents at $L$. By proportionality,
\begin{align}\label{eq lemma anon 0}
f(\boldsymbol{x}')=\frac{k^*L}{n}>f(\boldsymbol{x}).
\end{align}
 Now consider any subset  $G\subseteq N$ that contains $|k^*-k_1|>1$ agents who are located at $L$. Let 
 $\boldsymbol{x}''$ denote the profile obtained from $\boldsymbol{x}'$ by   shifting the $G$ agents' locations to $\hat{x}_G=f(\boldsymbol{x})$. We shall prove that 
\begin{align}\label{eq: proof w/o anon freemana}
f(\boldsymbol{x}'')=\hat{x}_G.
\end{align}
To see this, suppose that this is not the case. Then either 
\begin{align}\label{eq lemma anon 1}
 \hat{x}_G&<f(\boldsymbol{x}'') \qquad \text{or}\\
\label{eq lemma anon 2}
 f(\boldsymbol{x}'')&<\hat{x}_G.
\end{align}
 In the former case, all agents in $G$ have location strictly below $f(\boldsymbol{x}'')$---namely, $\hat{x}_G=f(\boldsymbol{x})$. Therefore, by the uncompromising property, if all agents in $G$   shift their location to $0$ in the profile $\boldsymbol{x}''$, then the facility location is unchanged and continues to be located at $f(\boldsymbol{x}'')$. Proportionality then requires that the facility then be located at $\frac{k_1L}{n}$ and, hence, $f(\boldsymbol{x}'')=\frac{k_1L}{n}$. But then~(\ref{eq lemma anon 1})  implies that 
 $\hat{x}_G<\frac{k_1L}{n}$, which in turn implies that
 $f(\boldsymbol{x})=\hat{x}_G<\frac{k_1L}{n}$---this contradicts~(\ref{eq 1}).  In the latter case,  all agents in $G$ have location strictly above $f(\boldsymbol{x}'')$---namely, $\hat{x}_G=f(\boldsymbol{x})$. Therefore, by the uncompromising property, if all agents in $G$   shift their location back to $L$ in the profile $\boldsymbol{x}''$, then the facility location is unchanged and, by proportionality, is located at $\frac{k^*L}{n}$. Hence, $f(\boldsymbol{x}'') =\frac{k^*L}{n}$. Using~(\ref{eq lemma anon 2}), this implies that $\frac{k^*L}{n}<f(\boldsymbol{x})$, which contradicts~(\ref{eq lemma anon 0}).

Now given~(\ref{eq: proof w/o anon freemana}), by the uncompromising property, shifting any agent with location at 0  (resp., 1) to any location weakly below  (resp., above) $\hat{x}_G$ must leave the facility's location unchanged. Therefore, for any profile  with exactly $k_0, k_1$ agents strictly below $\hat{x}_G$ and strictly above $\hat{x}_G$ and $n-k_0-k_1$ agents located at $\hat{x}_G$, the facility must be located at $\hat{x}_G$. But---since $\hat{x}_G=f(\boldsymbol{x})$---it is immediate that  any permutation of $\boldsymbol{x}$, say $\boldsymbol{x}_\sigma$, satisfies these 3 properties; hence, 
$$f(\boldsymbol{x}_\sigma)=f(\boldsymbol{x}).$$
We conclude that any mechanism that satisfies strategyproofness, proportionality, unanimity must also satisfy anonymity.
	\end{proof}

\subsection{Unanimity is Necessary for Theorem~\ref{thm: no need anon Freeman}.}\label{section:prop: cont needed}

\begin{proof}
We wish to prove that Theorem~\ref{thm: no need anon Freeman} does not hold if unanimity is removed. {It suffices to} consider the following mechanism for $n=2$
$$
f(\boldsymbol{x})=\begin{cases}
0 &\text{if $x_1, x_2\le L/4$,}\\
L &\text{if $x_1, x_2\ge 3L/4$,}\\
L/2 &\text{else.}
\end{cases}$$ 
This mechanism is clearly anonymous, satisfies proportionality and is not the Uniform Phantom mechanism. It remains to show that it is strategyproof. Using a symmetry argument, we focus on deviations by agent 1 without loss of generality. Suppose $f(\boldsymbol{x})=0$, then it must be that $x_1\le L/4$. But then agent 1 obtains the minimum possible distance to the facility (given $x_1$ and given the mechanism's range); hence, no deviation can  strictly decrease their distance. Suppose $f(\boldsymbol{x})=L/2$, then either $x_1\in (L/4, 3L/4)$ or $x_2\in (L/4,3L/4)$. In the former case, agent 1 obtains the minimum possible distance to the facility  (given $x_1$  and given the mechanism's range); hence, no deviation can  strictly decrease their distance. In the latter case, no deviation by agent 1 can change the facility location. Finally, suppose $f(\boldsymbol{x})=L$, then  it must be that $x_1\ge 3L/4$. But then agent 1 obtains the minimum possible distance to the facility (given $x_1$ and given the mechanism's range); hence, no deviation can  strictly decrease their distance.  Therefore, the mechanism is SP.
\end{proof}

\subsection{Proof of Proposition~\ref{prop: IFS bound cost}.}\label{section:prop: IFS bound cost}
\begin{proof}
By Theorem~\ref{char: IFS and SP}, a strategyproof, anonymous, unanimous mechanism satisfies IFS must be a Phantom mechanism with $n-1$ phantoms all contained in the interval $[\frac{1}{n},L-\frac{1}{n}]$. For such a Phantom mechanism, consider the location profile which places $n-1$ agents at $0$ and $1$ agent at the leftmost phantom $p_1$. The Phantom mechanism places the facility at $p_1$, leading to $(n-1)p_1$ total cost, and the optimal total cost is $p_1$, achieved by placing the facility at $0$. The total cost approximation ratio is therefore at least $n-1$. Now, any Pareto optimal mechanism has a total cost approximation ratio of at most $n-1$. To see this, note that for any agent location profile $\boldsymbol{x}$ where the distance between the leftmost and rightmost agents is $c$, we must have $\Psi^*(\boldsymbol{x})\geq c$, and $\Psi(f(\boldsymbol{x}))\leq (n-1)c$. As a Phantom mechanism with $n-1$ phantoms is Pareto optimal, we see that a Phantom mechanism with $n-1$ phantoms all contained in the interval $[\frac{1}{n},L-\frac{1}{n}]$ has a total cost approximation ratio of $n-1$, and thus the proposition statement follows.
\end{proof}
\subsection{Proof of Lemma~\ref{lemma: IFS bound}.}\label{section:lemma: IFS bound}

\begin{proof}
We wish to prove that a mechanism $f$ that satisfies IFS has welfare approximation of at least $1+\frac{n-2}{n^2-2n+2}$. To this end, suppose $f$ satisfies IFS and consider the profile of locations $\boldsymbol{x}\in \{0,L\}^n$ that places $n-1$ agents at $0$. IFS requires that  $f(\boldsymbol{x})\geq\frac{L}{n}$ (and $f(\boldsymbol{x})\leq L(1-\frac{1}{n})$), so any IFS mechanism has welfare of at most 
$$\Phi(f(\boldsymbol{x}))\leq (n-1)L(1-\frac{1}{n})+\frac{L}{n}=\frac{L(n-1)^2 +L}{n}.$$ 
However, for this instance, the {welfare-optimal} welfare is $\Phi^*(\boldsymbol{x})=L(n-1)$ (obtained by locating the facility at the median location, $0$). Therefore, the approximation ratio of $f$ is at least
$$\frac{\Phi^*(\boldsymbol{x})}{\Phi(f(\boldsymbol{x}))}=\frac{nL(n-1)}{ L(n-1)^2 +L }=1+\frac{n-2}{(n-1)^2+1}.$$
\end{proof}

\subsection{Proof of Theorem~\ref{thm: IFS bound}.}\label{section:thm: IFS bound}

\begin{proof}
We wish to prove that among all IFS mechanisms, the Constrained Median mechanism provides the best welfare approximation guarantee i.e.,  it achieves the approximation ratio in Lemma~\ref{lemma: IFS bound}. Let $f_{\text{CM}}$ denote the Constrained Median mechanism.  We shall prove that for any location profile $\boldsymbol{x}\in [0,L]^n$ there exists some profile $\tilde{\boldsymbol{x}}\in \{0,L\}^n$ such that 
\begin{align}
\label{equation: approx ratio restricted equality0}
\frac{\Phi^*(\tilde{\boldsymbol{x}})}{\Phi(f_{\text{CM}}(\tilde{\boldsymbol{x}}))}  &\ge \frac{\Phi^*(\boldsymbol{x})}{\Phi(f_{\text{CM}}(\boldsymbol{x}))}, 
\end{align}
which implies that
\[\max_{\boldsymbol{x}\in [0,L]^n}\frac{\Phi^*(\boldsymbol{x})}{\Phi(f_{\text{CM}}(\boldsymbol{x}))}=\max_{\boldsymbol{x}\in \{0,L\}^n}\frac{\Phi^*(\boldsymbol{x})}{\Phi(f_{\text{CM}}(\boldsymbol{x}))}.\]
We begin by noting that, whenever $f_{\text{CM}}(\boldsymbol{x})\in (L/n, L(1-1/n))$, the facility location coincides with the median location and $f_{\text{CM}}$ obtains the maximum welfare.  Thus, we can restrict our attention to profiles such that $f_{\text{CM}}(\boldsymbol{x})\notin (L/n, L(1-1/n))$.  We proceed to prove~(\ref{equation: approx ratio restricted equality0}) by considering a sequence of  profiles that modify $\boldsymbol{x}$ into some profile $\tilde{\boldsymbol{x}}\in \{0,L\}^n$ such that each modified profile guarantees a weakly higher  welfare  approximation ratio.

 Let the agent labels be ordered such that $x_1\le \ldots \le x_n$; let $i={med}$ denote the median agent. Without loss of generality, suppose $f_{\text{CM}}(\boldsymbol{x})\in [0, L/n]$. This implies that the median agent is weakly below $f_{\text{CM}}(\boldsymbol{x})$, i.e.,   $x_{med}\le f_{\text{CM}}(\boldsymbol{x})$. {To assist with visualizing the proof technique, we provide a running example with $n=5$ agents. Figure~\ref{figure:thm5x} illustrates a profile $\boldsymbol{x}$ such that $x_{med}\le f_{\text{CM}}(\boldsymbol{x})$; in particular, $x_{med}=x_3$ and $ f_{\text{CM}}(\boldsymbol{x})=L/5$. }

   \begin{figure}[H]
 	  	 \begin{center}  
   		      	             \begin{tikzpicture}[scale=1]
   			   	      	                 \centering
   	   	      	                 \draw[-] (0,0) -- (12,0);
								 
   								   \draw[-] (0,0) -- (0,0.25);
   								    \draw[-] (12,0) -- (12,0.25);
   									      \draw[-] (2,0) -- (2,0.25);
     \draw[-] (10,0) -- (10,0.25);
  
     \draw (0,-.4) node(c){\small $0$};
       \draw (2,-.4) node(c){\small $L/5$};
   			 	 \draw (10,-.4) node(c){\small $4L/5$};
       \draw (12,-.4) node(c){\small $L$};
       
	   		  	  	  \draw (0,0.5) node(c){\small $x_1$};

   	  \draw (0.5,0.5) node(c){\small $x_2$};
   	  \draw (1,0.5) node(c){\small $x_3$};
	     				   	   \draw (1.5,0.5) node(c){\small $x_4$};

   	  	  \draw (4,0.5) node(c){\small $x_5$};

   			   	      	  \end{tikzpicture}
   			   	       	\end{center}
   			   	      	 \caption{Running example. Profile $\boldsymbol{x}$}
   			   	      	\label{figure:thm5x}
   			   	      	\end{figure}
 
 First,  consider the modified profile $\boldsymbol{x}'$ such that $x_i'=0$ for $i\in N':=\{i \ : \ i<{med}\}$ and $x_i'=x_i$ for all $i\notin N'$. {Applying this operation to the running example illustrated in Figure~\ref{figure:thm5x}, we obtain the profile illustrated in Figure~\ref{figure:thm5x'}.   }

   \begin{figure}[H]
 	  	 \begin{center}  
   		      	             \begin{tikzpicture}[scale=1]
   			   	      	                 \centering
   	   	      	                 \draw[-] (0,0) -- (12,0);
								 
   								   \draw[-] (0,0) -- (0,0.25);
   								    \draw[-] (12,0) -- (12,0.25);
   									      \draw[-] (2,0) -- (2,0.25);
     \draw[-] (10,0) -- (10,0.25);
  
     \draw (0,-.4) node(c){\small $0$};
       \draw (2,-.4) node(c){\small $L/5$};
   			 	 \draw (10,-.4) node(c){\small $4L/5$};
       \draw (12,-.4) node(c){\small $L$};
       
	   		  	  	  \draw (0,0.5) node(c){\small $x_1'$};

   	  \draw (0,1) node(c){\small $x_2'$};
   	  \draw (1,0.5) node(c){\small $x_3'$};
	     				   	   \draw (1.5,0.5) node(c){\small $x_4'$};

   	  	  \draw (4,0.5) node(c){\small $x_5'$};

   			   	      	  \end{tikzpicture}
   			   	       	\end{center}
   			   	      	 \caption{Running example. Profile $\boldsymbol{x}'$}
   			   	      	\label{figure:thm5x'}
   			   	      	\end{figure}
 
   In this modified profile, {we have moved all agents strictly left of the median agent to $0$, so} neither the {welfare-optimal} (median) location nor the facility location under $f_{\text{CM}}$ changes. Hence, relative to $\Phi^*(\boldsymbol{x})$ and $\Phi(f_{\text{CM}}(\boldsymbol{x}))$, the optimal  welfare, $\Phi^*(\boldsymbol{x}')$, and the  welfare provided by $f_{\text{CM}}$, $\Phi(f_{\text{CM}}(\boldsymbol{x}'))$, decrease by the same amount---namely, $\sum_{i\in N'} x_i\ge 0$. We conclude that
\begin{align*}
\frac{\Phi^*(\boldsymbol{x}')}{\Phi(f_{\text{CM}}(\boldsymbol{x}'))}
&= \frac{\Phi^*(\boldsymbol{x})-\sum_{i\in N'} x_i}{\Phi(f_{\text{CM}}(\boldsymbol{x}))-\sum_{i\in N'} x_i}\ge \frac{\Phi^*(\boldsymbol{x})}{\Phi(f_{\text{CM}}(\boldsymbol{x}))},
\end{align*}
where the final inequality follows because $(x-a)/(y-a)\ge x/y$ for any $a\ge 0$ and $0<y\le x$.

Next we consider the  modified profile $\boldsymbol{x}''$ such that $x_{med}''=0$ and $x_i''=x_i'$ for all $i\neq med$. {Applying this operation to the running example illustrated in Figure~\ref{figure:thm5x'}, we obtain the profile illustrated in Figure~\ref{figure:thm5x''}.   }

   \begin{figure}[H]
 	  	 \begin{center}  
   		      	             \begin{tikzpicture}[scale=1]
   			   	      	                 \centering
   	   	      	                 \draw[-] (0,0) -- (12,0);
								 
   								   \draw[-] (0,0) -- (0,0.25);
   								    \draw[-] (12,0) -- (12,0.25);
   									      \draw[-] (2,0) -- (2,0.25);
     \draw[-] (10,0) -- (10,0.25);
  
     \draw (0,-.4) node(c){\small $0$};
       \draw (2,-.4) node(c){\small $L/5$};
   			 	 \draw (10,-.4) node(c){\small $4L/5$};
       \draw (12,-.4) node(c){\small $L$};
       
	   		  	  	  \draw (0,0.5) node(c){\small $x_1''$};

   	  \draw (0,1) node(c){\small $x_2''$};
   	  \draw (0,1.5) node(c){\small $x_3''$};
	     				   	   \draw (1.5,0.5) node(c){\small $x_4''$};

   	  	  \draw (4,0.5) node(c){\small $x_5''$};

   			   	      	  \end{tikzpicture}
   			   	       	\end{center}
   			   	      	 \caption{Running example. Profile $\boldsymbol{x}''$}
   			   	      	\label{figure:thm5x''}
   			   	      	\end{figure}

 In this modified profile,  the welfare-optimal (median) location moves from $x_{med}$ to 0, so the facility location under $f_{\text{CM}}$ remains unchanged, i.e., $f_{\text{CM}}(\boldsymbol{x}'')=f(\boldsymbol{x}')$. Hence, relative to  $\Phi^*(\boldsymbol{x}')$, the optimal welfare, $\Phi^*(\boldsymbol{x}'')$, decreases by $x_{med}$ if $n$ is even and decreases by $0$ otherwise;  relative to  $\Phi(f_{\text{CM}}(\boldsymbol{x}'))$, the   welfare  under $f_{\text{CM}}$, $\Phi(f_{\text{CM}}(\boldsymbol{x}''))$, decreases by $x_{med}$. {Defining the indicator function $ \mathbb{I}_{n \text{ even.}}$ as $1$ if $n$ is even and $0$ otherwise}, we conclude that  
\begin{align*}
\frac{\Phi^*(\boldsymbol{x}'')}{\Phi(f_{\text{CM}}(\boldsymbol{x}''))}
&= \frac{\Phi^*(\boldsymbol{x}')- x_{med}\mathbb{I}_{n \text{ even.}}}{\Phi(f_{\text{CM}}(\boldsymbol{x}'))- x_{med}}\ge \frac{\Phi^*(\boldsymbol{x}') }{\Phi(f_{\text{CM}}(\boldsymbol{x}')) } \ge \frac{\Phi^*(\boldsymbol{x}) }{\Phi(f_{\text{CM}}(\boldsymbol{x})) }. 
\end{align*}
Now either $x_n\ge L/n$ or $x_n<L/n$. Suppose the former case holds, then 
$$f_{\text{CM}}(\boldsymbol{x})=L/n=f_{\text{CM}}(\boldsymbol{x}')=f_{\text{CM}}(\boldsymbol{x}'').$$
Consider the modified profile $\boldsymbol{x}'''\in \{0,L\}^n$ such that $x_i'''=L$ for all $i\in N''':=\{ i \ : x_i''\ge L/n\}$ and $x_i'''=0$ for all  $i\notin N'''$.  {Applying this operation to the running example illustrated in Figure~\ref{figure:thm5x''}, we obtain the profile illustrated in Figure~\ref{figure:thm5x'''}.   }

   \begin{figure}[H]
 	  	 \begin{center}  
   		      	             \begin{tikzpicture}[scale=1]
   			   	      	                 \centering
   	   	      	                 \draw[-] (0,0) -- (12,0);
								 
   								   \draw[-] (0,0) -- (0,0.25);
   								    \draw[-] (12,0) -- (12,0.25);
   									      \draw[-] (2,0) -- (2,0.25);
     \draw[-] (10,0) -- (10,0.25);
  
     \draw (0,-.4) node(c){\small $0$};
       \draw (2,-.4) node(c){\small $L/5$};
   			 	 \draw (10,-.4) node(c){\small $4L/5$};
       \draw (12,-.4) node(c){\small $L$};
       
	   		  	  	  \draw (0,0.5) node(c){\small $x_1'''$};

   	  \draw (0,1) node(c){\small $x_2'''$};
   	  \draw (0,1.5) node(c){\small $x_3'''$};
	     				   	   \draw (0,2) node(c){\small $x_4'''$};

   	  	  \draw (12,0.5) node(c){\small $x_5'''$};

   			   	      	  \end{tikzpicture}
   			   	       	\end{center}
   			   	      	 \caption{Running example. Profile $\boldsymbol{x}'''$}
   			   	      	\label{figure:thm5x'''}
   			   	      	\end{figure}

 In this modified profile, {we have moved all agent locations that were weakly right of $L/n$ in $\boldsymbol{x}''$ to $L$ and all other agents' locations are shifted to $0$.} Under $\boldsymbol{x}'''$, the welfare-optimal (median) location remains unchanged (at $x_{med}''=0$) and the facility location under $f_{\text{CM}}$ remains at $L/n$.  We conclude that 
\begin{align*}
\frac{\Phi^*(\boldsymbol{x}''')}{\Phi(f_{\text{CM}}(\boldsymbol{x}'''))}
&= \frac{\Phi^*(\boldsymbol{x}'')-\sum_{i\in N'''} (L-x_i'')+\sum_{i\notin N'''} x_i''}{\Phi(f_{\text{CM}}(\boldsymbol{x}''))-\sum_{i\in N'''} (L-x_i'')-\sum_{i\notin N'''} x_i''}\ge \frac{\Phi^*(\boldsymbol{x}'') }{\Phi(f_{\text{CM}}(\boldsymbol{x}'')) }\ge
   \frac{\Phi^*(\boldsymbol{x}) }{\Phi(f_{\text{CM}}(\boldsymbol{x})) }.
\end{align*}
Therefore, there exists $\tilde{\boldsymbol{x}}\in \{0,L\}^n$---namely, $\boldsymbol{x}'''$---with weakly higher  welfare  approximation ratio than $\boldsymbol{x}$.

Finally, suppose the latter case, $x_n<L/n$, holds. In this case, 
$$f_{\text{CM}}(\boldsymbol{x})=x_n=f_{\text{CM}}(\boldsymbol{x}')=f_{\text{CM}}(\boldsymbol{x}'')<L/n.$$
Consider the modified profile  $\boldsymbol{x}''''$ such that $x_n''''=L/n$  and $x_i''''=x_i''$ otherwise. In this modified profile, {we have moved the last agent $x_n''$ to $L/n$, so} the {welfare-optimal} (median) location remains unchanged (at $x_{med}''=0$) and the facility location under $f_{\text{CM}}$ shifts to $L/n$, i.e., $f_{\text{CM}}(\boldsymbol{x}'''')=L/n$. We conclude that  
\begin{align*}
\frac{\Phi^*(\boldsymbol{x}'''')}{\Phi(f_{\text{CM}}(\boldsymbol{x}''''))}
&= \frac{\Phi^*(\boldsymbol{x}'')-(L/n-x_n)}{\Phi(f_{\text{CM}}(\boldsymbol{x}''))-(n-1)(L/n-x_n)}\ge  \frac{\Phi^*(\boldsymbol{x}'')}{\Phi(f_{\text{CM}}(\boldsymbol{x}''))}.
\end{align*}
Now the same steps from the former case can be used to show that there exists $\tilde{\boldsymbol{x}}\in \{0,L\}^n$ with weakly higher welfare approximation ratio than $\boldsymbol{x}$. Therefore,~(\ref{equation: approx ratio restricted equality0}) holds. 

It is straightforward to calculate the maximum  welfare  approximation ratio among profiles $\tilde{\boldsymbol{x}}\in \{0,L\}^n$. The maximum is attained when $ \tilde{\boldsymbol{x}}$ has $(n-1)$ agents at 0 and $1$ agent at $L$, which provides the required  welfare  approximation ratio (see Proof of Lemma~\ref{lemma: IFS bound}).
 \end{proof}

\subsection{Proof of Lemma~\ref{lemma: UFS bound}.}\label{section:lemma: UFS bound}

\begin{proof}
We wish to prove that any mechanism satisfying UFS (or proportionality or PF) has a welfare approximation of at least~(\ref{eq: lemma UFS lower bound}). To this end,
suppose $f$ satisfies UFS. Consider the agent location profile $\boldsymbol{x}\in \{0,L\}^n$ that has $k\le n/2$ agents at $L$.  The optimal  welfare $\Phi^*(\boldsymbol{x})=L(n-k)$ is obtained by placing the facility at the median location $0$. UFS requires that $f(\boldsymbol{x})=\frac{kL}{n}$, which provides welfare $\Phi(f(\boldsymbol{x}))=\frac{L(k^2+(n-k)^2)}{n}$. Therefore, the  welfare approximation ratio is
\begin{align*}
\frac{\Phi^*(\boldsymbol{x})}{\Phi(f(\boldsymbol{x}))}&=\frac{n(n-k)}{k^2+(n-k)^2}.
\end{align*}
Maximizing the above expression with respect to $k\in \mathbb{N} \ :\  0\le k\le n/2 $ provides the  welfare  approximation bound in the lemma statement. {Defining $r:=\frac{k}{n}$, this ratio is equal to
\[\frac{\Phi^*(\boldsymbol{x})}{\Phi(f(\boldsymbol{x}))}=\frac{1-r}{2r^2-2r+1}.\]
The derivative of this expression with respect to $r$ is $\frac{2r^2-4r+1}{(2r^2-2r+1)^2}$, which is equal to $0$ when $r=\frac{2-\sqrt{2}}{2}$ or $r=\frac{2+\sqrt{2}}{2}$. We ignore the latter as $k$ cannot exceed $n$, and we note that $r=\frac{2-\sqrt{2}}{2}$ is a maximum point as the derivative is positive for $r\in [0,\frac{2-\sqrt{2}}{2})$ and negative for $r\in (\frac{2-\sqrt{2}}{2},1]$. We therefore deduce that $\frac{\Phi^*(\boldsymbol{x})}{\Phi(f(\boldsymbol{x}))}$ is maximized when $\frac{k}{n}=\frac{2-\sqrt{2}}{2}$, providing  welfare  approximation ratio  $\frac{\sqrt{2}+1}{2}$.} This approximation ratio can be achieved asymptotically as $n \rightarrow \infty$.
\end{proof}

\subsection{Proof of Theorem~\ref{thm: UFS bound}.}\label{section:thm: UFS bound}

\begin{proof}
We wish to prove that among all UFS (or proportional or PF) mechanisms, the Uniform Phantom mechanism provides the best  welfare  approximation guarantee, i.e.,  it achieves the approximation ratio in Lemma~\ref{lemma: UFS bound}. To this end, let $f_{\text{Unif}}$ denote the Uniform Phantom mechanism. We prove that for any location profile $\boldsymbol{x}\in [0,L]^n$ there exists some profile $\tilde{\boldsymbol{x}}\in \{0,L\}^n$ such that
\begin{align}\label{equation: funif approx}
\frac{\Phi^*(\tilde{\boldsymbol{x}})}{\Phi(f_{\text{Unif}}(\tilde{\boldsymbol{x}}))}  &\ge \frac{\Phi^*(\boldsymbol{x})}{\Phi(f_{\text{Unif}}(\boldsymbol{x}))}.
\end{align}
This implies that
\[\max_{\boldsymbol{x}\in [0,L]^n}\frac{\Phi^*(\boldsymbol{x})}{\Phi(f_{\text{Unif}}(\boldsymbol{x}))}=\max_{\boldsymbol{x}\in \{0,L\}^n}\frac{\Phi^*(\boldsymbol{x})}{\Phi(f_{\text{Unif}}(\boldsymbol{x}))}.\]

Let the agent labels be ordered such that $x_1\le \ldots \le x_n$; let $i={med}$ denote the median agent.  Suppose without loss of generality that  $\boldsymbol{x} \ : \ x_{med}< f_{\text{Unif}}(\boldsymbol{x})$; if $x_{med}=f_{\text{Unif}}(\boldsymbol{x})$, then~(\ref{equation: funif approx}) is trivially satisfied.  {To assist with visualizing the proof technique, we provide a running example with $n=6$ agents. Figure~\ref{figure:thm6x} illustrates a profile $\boldsymbol{x}$ such that $x_{med}< f_{\text{Unif}}(\boldsymbol{x})$; in particular, $x_{med}=x_3$ and $ f_{\text{Unif}}(\boldsymbol{x})=x_5$. }

   \begin{figure}[H]
 	  	 \begin{center}  
   		      	             \begin{tikzpicture}[scale=1]
   			   	      	                 \centering
   	   	      	                 \draw[-] (0,0) -- (12,0);
								 
   								   \draw[-] (0,0) -- (0,0.25);
								   							   \draw[-] (4,0) -- (4,0.25);
															   							   \draw[-] (6,0) -- (6,0.25);
																						   							   \draw[-] (8,0) -- (8,0.25);
   								    \draw[-] (12,0) -- (12,0.25);
   									      \draw[-] (2,0) -- (2,0.25);
     \draw[-] (10,0) -- (10,0.25);
  
     \draw (0,-.4) node(c){\small $0$};
       \draw (2,-.4) node(c){\small $L/6$};
              \draw (4,-.4) node(c){\small $2L/6$};
                            \draw (6,-.4) node(c){\small $3L/6$};
                                                        \draw (8,-.4) node(c){\small $4L/6$};
   			 	 \draw (10,-.4) node(c){\small $5L/6$};
       \draw (12,-.4) node(c){\small $L$};
       
	   		  	  	  \draw (0,0.5) node(c){\small $x_1$};

   	  \draw (0.5,0.5) node(c){\small $x_2$};
   	  \draw (1,0.5) node(c){\small $x_3$};
	     				   	   \draw (1.5,0.5) node(c){\small $x_4$};
						    \draw (3,0.5) node(c){\small $x_5$};

   	  	  \draw (4,0.5) node(c){\small $x_6$};

   			   	      	  \end{tikzpicture}
   			   	       	\end{center}
   			   	      	 \caption{Running example. Profile $\boldsymbol{x}$}
   			   	      	\label{figure:thm6x}
   			   	      	\end{figure}

First, consider the modified profile $\boldsymbol{x}'$ such that $x_i'=L$  for all $i\in N':=\{i \ :\ f_{\text{Unif}}(\boldsymbol{x})<x_i\}$, $x_i'=0$ for all $i\in N'':=\{i:i<med\}$, and $x_i'=x_i$ for all $i\notin N'\cup N''$---note that $N'\cap N''=\emptyset$. {Applying this operation to the running example illustrated in Figure~\ref{figure:thm6x}, we obtain the profile illustrated in Figure~\ref{figure:thm6x'}.   }

   \begin{figure}[H]
 	  	 \begin{center}  
   		      	             \begin{tikzpicture}[scale=1]
   			   	      	                 \centering
   	   	      	                 \draw[-] (0,0) -- (12,0);
								 
   								   \draw[-] (0,0) -- (0,0.25);
								   							   \draw[-] (4,0) -- (4,0.25);
															   							   \draw[-] (6,0) -- (6,0.25);
																						   							   \draw[-] (8,0) -- (8,0.25);
   								    \draw[-] (12,0) -- (12,0.25);
   									      \draw[-] (2,0) -- (2,0.25);
     \draw[-] (10,0) -- (10,0.25);
  
     \draw (0,-.4) node(c){\small $0$};
       \draw (2,-.4) node(c){\small $L/6$};
              \draw (4,-.4) node(c){\small $2L/6$};
                            \draw (6,-.4) node(c){\small $3L/6$};
                                                        \draw (8,-.4) node(c){\small $4L/6$};
   			 	 \draw (10,-.4) node(c){\small $5L/6$};
       \draw (12,-.4) node(c){\small $L$};
       
	   		  	  	  \draw (0,0.5) node(c){\small $x_1'$};

   	  \draw (0,1) node(c){\small $x_2'$};
   	  \draw (1,0.5) node(c){\small $x_3'$};
	     				   	   \draw (1.5,0.5) node(c){\small $x_4'$};
						    \draw (3,0.5) node(c){\small $x_5'$};

   	  	  \draw (12,0.5) node(c){\small $x_6'$};

   			   	      	  \end{tikzpicture}
   			   	       	\end{center}
   			   	      	 \caption{Running example. Profile $\boldsymbol{x}'$}
   			   	      	\label{figure:thm6x'}
   			   	      	\end{figure}

 In this modified profile, {we have moved all agents with location strictly to the right of the Uniform Phantom mechanism location to $L$, and all agents strictly left of the median to 0.} {Under $\boldsymbol{x}'$,} neither the {welfare-optimal} (median) location nor the facility location under $f_{\text{Unif}}$ changes. Therefore, relative to $\Phi^*(\boldsymbol{x})$ and $\Phi(f_{\text{Unif}}(\boldsymbol{x}))$, the optimal  welfare, $\Phi^*(\boldsymbol{x}')$, and the  welfare under $f$, $\Phi(f_{\text{Unif}}(\boldsymbol{x}'))$, decrease by the same amount---namely, $\sum_{i\in N'}(L-x_i)+\sum_{i\in N''}x_i \geq 0$. We conclude that
\begin{align*}
\frac{\Phi^*(\boldsymbol{x}')}{\Phi(f_{\text{Unif}}(\boldsymbol{x}'))}
&= \frac{\Phi^*(\boldsymbol{x})-\sum_{i\in N'} (L-x_i)-\sum_{i\in N''} x_i}{\Phi(f_{\text{Unif}}(\boldsymbol{x}))-\sum_{i\in N'} (L-x_i)-\sum_{i\in N''} x_i}\ge \frac{\Phi^*(\boldsymbol{x})}{\Phi(f_{\text{Unif}}(\boldsymbol{x}))}.
\end{align*}

Next we consider the modified profile $\boldsymbol{x}''$ such that $x_{med}''=0$ and $x_i''=x_i'$ for all $i\neq med$.  {Applying this operation to the running example illustrated in Figure~\ref{figure:thm6x'}, we obtain the profile illustrated in Figure~\ref{figure:thm6x''}.   }

   \begin{figure}[H]
 	  	 \begin{center}  
   		      	             \begin{tikzpicture}[scale=1]
   			   	      	                 \centering
   	   	      	                 \draw[-] (0,0) -- (12,0);
								 
   								   \draw[-] (0,0) -- (0,0.25);
								   							   \draw[-] (4,0) -- (4,0.25);
															   							   \draw[-] (6,0) -- (6,0.25);
																						   							   \draw[-] (8,0) -- (8,0.25);
   								    \draw[-] (12,0) -- (12,0.25);
   									      \draw[-] (2,0) -- (2,0.25);
     \draw[-] (10,0) -- (10,0.25);
  
     \draw (0,-.4) node(c){\small $0$};
       \draw (2,-.4) node(c){\small $L/6$};
              \draw (4,-.4) node(c){\small $2L/6$};
                            \draw (6,-.4) node(c){\small $3L/6$};
                                                        \draw (8,-.4) node(c){\small $4L/6$};
   			 	 \draw (10,-.4) node(c){\small $5L/6$};
       \draw (12,-.4) node(c){\small $L$};
       
	   		  	  	  \draw (0,0.5) node(c){\small $x_1''$};

   	  \draw (0,1) node(c){\small $x_2''$};
   	  \draw (0,1.5) node(c){\small $x_3''$};
	     				   	   \draw (1.5,0.5) node(c){\small $x_4''$};
						    \draw (3,0.5) node(c){\small $x_5''$};

   	  	  \draw (12,0.5) node(c){\small $x_6''$};

   			   	      	  \end{tikzpicture}
   			   	       	\end{center}
   			   	      	 \caption{Running example. Profile $\boldsymbol{x}''$}
   			   	      	\label{figure:thm6x''}
   			   	      	\end{figure}

In this modified profile, the {welfare-optimal} (median) location moves from $x_{med}$ to 0 and the facility location under $f_{\text{Unif}}$ remains unchanged, i.e., $f_{\text{Unif}}(\boldsymbol{x}'')=f_{\text{Unif}}(\boldsymbol{x}')$. Hence,   relative to $\Phi^*(\boldsymbol{x}')$, the optimal  welfare, $\Phi^*(\boldsymbol{x}'')$, decreases by $x_{med}$ if $n$ is even and decreases by $0$ otherwise; relative  to $\Phi(f_{\text{Unif}}(\boldsymbol{x}'))$, the   welfare under $f_{\text{Unif}}$, $\Phi(f_{\text{Unif}}(\boldsymbol{x}''))$, decreases by $x_{med}$. We conclude that 
\begin{align*}
\frac{\Phi^*(\boldsymbol{x}'')}{\Phi(f_{\text{Unif}}(\boldsymbol{x}''))}
&= \frac{\Phi^*(\boldsymbol{x}')- x_{med}\mathbb{I}_{n \text{ even.}}}{\Phi(f_{\text{Unif}}(\boldsymbol{x}'))- x_{med}}\ge \frac{\Phi^*(\boldsymbol{x}') }{\Phi(f_{\text{Unif}}(\boldsymbol{x}')) } \ge \frac{\Phi^*(\boldsymbol{x}) }{\Phi(f_{\text{Unif}}(\boldsymbol{x})) }. 
\end{align*}
Now consider the modified profile $\boldsymbol{x}'''$ such that $x_i'''=0$ for all $i\in N''':=\{i:x_i''< f_{\text{Unif}}(\boldsymbol{x})\}$ and $x_i'''=x_i''$ for all $i\notin N'''$. {Applying this operation to the running example illustrated in Figure~\ref{figure:thm6x''}, we obtain the profile illustrated in Figure~\ref{figure:thm6x'''}.   }

   \begin{figure}[H]
 	  	 \begin{center}  
   		      	             \begin{tikzpicture}[scale=1]
   			   	      	                 \centering
   	   	      	                 \draw[-] (0,0) -- (12,0);
								 
   								   \draw[-] (0,0) -- (0,0.25);
								   							   \draw[-] (4,0) -- (4,0.25);
															   							   \draw[-] (6,0) -- (6,0.25);
																						   							   \draw[-] (8,0) -- (8,0.25);
   								    \draw[-] (12,0) -- (12,0.25);
   									      \draw[-] (2,0) -- (2,0.25);
     \draw[-] (10,0) -- (10,0.25);
  
     \draw (0,-.4) node(c){\small $0$};
       \draw (2,-.4) node(c){\small $L/6$};
              \draw (4,-.4) node(c){\small $2L/6$};
                            \draw (6,-.4) node(c){\small $3L/6$};
                                                        \draw (8,-.4) node(c){\small $4L/6$};
   			 	 \draw (10,-.4) node(c){\small $5L/6$};
       \draw (12,-.4) node(c){\small $L$};
       
	   		  	  	  \draw (0,0.5) node(c){\small $x_1'''$};

   	  \draw (0,1) node(c){\small $x_2'''$};
   	  \draw (0,1.5) node(c){\small $x_3'''$};
	     				   	   \draw (0,2) node(c){\small $x_4'''$};
						    \draw (3,0.5) node(c){\small $x_5'''$};

   	  	  \draw (12,0.5) node(c){\small $x_6'''$};

   			   	      	  \end{tikzpicture}
   			   	       	\end{center}
   			   	      	 \caption{Running example. Profile $\boldsymbol{x}'''$}
   			   	      	\label{figure:thm6x'''}
   			   	      	\end{figure}

In this modified profile, {we move all agents strictly left of the Uniform Phantom mechanism's facility location to $0$, so} neither the {welfare-optimal} (median) location of 0, nor the facility location under $f_{\text{Unif}}$ changes. Hence, relative to $\Phi^*(\boldsymbol{x}'')$, the optimal  welfare, $\Phi^*(\boldsymbol{x}'')$, increases by $\sum_{i\in N'''} x_i''$; relative  to $\Phi(f_{\text{Unif}}(\boldsymbol{x}''))$, the   welfare under $f_{\text{Unif}}$, $\Phi(f_{\text{Unif}}(\boldsymbol{x}''))$, decreases by $\sum_{i\in N'''} x_i''$. We conclude that
\begin{align*}
\frac{\Phi^*(\boldsymbol{x}''')}{\Phi(f_{\text{Unif}}(\boldsymbol{x}'''))}
&= \frac{\Phi^*(\boldsymbol{x}'')+\sum_{i\in N'''} x_i''}{\Phi(f_{\text{Unif}}(\boldsymbol{x}''))-\sum_{i\in N'''} x_i''}\ge \frac{\Phi^*(\boldsymbol{x}'') }{\Phi(f_{\text{Unif}}(\boldsymbol{x}'')) }\ge  \frac{\Phi^*(\boldsymbol{x}) }{\Phi(f_{\text{Unif}}(\boldsymbol{x})) }.
\end{align*}
Lastly, consider the modified profile $\boldsymbol{x}''''$ such that $x_i''''=L$ for all $i\in N''''=\{i \ : \  f_{\text{Unif}}(\boldsymbol{x'''}) \le x_i\}$ and $x_i''''=0$ for all $i\notin N''''$. {Applying this operation to the running example illustrated in Figure~\ref{figure:thm6x'''}, we obtain the profile illustrated in Figure~\ref{figure:thm6x''''}. In  Figure~\ref{figure:thm6x''''}, the Uniform Phantom mechanism's  location increases to $2L/6$.}

   \begin{figure}[H]
 	  	 \begin{center}  
   		      	             \begin{tikzpicture}[scale=1]
   			   	      	                 \centering
   	   	      	                 \draw[-] (0,0) -- (12,0);
								 
   								   \draw[-] (0,0) -- (0,0.25);
								   							   \draw[-] (4,0) -- (4,0.25);
															   							   \draw[-] (6,0) -- (6,0.25);
																						   							   \draw[-] (8,0) -- (8,0.25);
   								    \draw[-] (12,0) -- (12,0.25);
   									      \draw[-] (2,0) -- (2,0.25);
     \draw[-] (10,0) -- (10,0.25);
  
     \draw (0,-.4) node(c){\small $0$};
       \draw (2,-.4) node(c){\small $L/6$};
              \draw (4,-.4) node(c){\small $2L/6$};
                            \draw (6,-.4) node(c){\small $3L/6$};
                                                        \draw (8,-.4) node(c){\small $4L/6$};
   			 	 \draw (10,-.4) node(c){\small $5L/6$};
       \draw (12,-.4) node(c){\small $L$};
       
	   		  	  	  \draw (0,0.5) node(c){\small $x_1''''$};

   	  \draw (0,1) node(c){\small $x_2''''$};
   	  \draw (0,1.5) node(c){\small $x_3''''$};
	     				   	   \draw (0,2) node(c){\small $x_4''''$};
						    \draw (12,1) node(c){\small $x_5''''$};

   	  	  \draw (12,0.5) node(c){\small $x_6''''$};

   			   	      	  \end{tikzpicture}
   			   	       	\end{center}
   			   	      	 \caption{Running example. Profile $\boldsymbol{x}''''$}
   			   	      	\label{figure:thm6x''''}
   			   	      	\end{figure}

Under this modified profile, {we have moved all agents weakly right of the Uniform Phantom mechanism's location to $L$, so} the {welfare-optimal} (median) location does not change;  the  facility location under $f_{\text{Unif}}$ moves to a (weakly) higher location, i.e., $f_{\text{Unif}}(\boldsymbol{x}'''') \ : \ f_{\text{Unif}}(\boldsymbol{x}''')\le f_{\text{Unif}}(\boldsymbol{x}'''')$.  

Relative to $\Phi^*(\boldsymbol{x}''')$, the optimal  welfare, $\Phi^*(\boldsymbol{x}'''')$, decreases by $\sum_{i\in N''''}(L-x_i''')$. Relative to $\Phi(f_{\text{Unif}}(\boldsymbol{x}'''))$, the welfare under $f_{\text{Unif}}$, $\Phi(f_{\text{Unif}}(\boldsymbol{x}''''))$ also decreases by $\sum_{i\in N''''}(L-x_i''')$ due to the movement in agents in $N''''$. In addition, $\Phi(f_{\text{Unif}}(\boldsymbol{x}''''))$ decreases due to the movement in the facility location: this follows because the number of agents at location 0 is weakly higher than the number of agents at location $L$. Let this additional decrease in $\Phi(f_{\text{Unif}}(\boldsymbol{x}''''))$ be denoted by $\Delta>0$.   We conclude that
\begin{align*}
\frac{\Phi^*(\boldsymbol{x}'''')}{\Phi(f_{\text{Unif}}(\boldsymbol{x}''''))}
&= \frac{\Phi^*(\boldsymbol{x}''')-\sum_{i\in N''''}(L-x_i''')}{\Phi(f_{\text{Unif}}(\boldsymbol{x}'''))-\sum_{i\in N''''}(L-x_i''')-\Delta}\ge \frac{\Phi^*(\boldsymbol{x}''') }{\Phi(f_{\text{Unif}}(\boldsymbol{x}''')) }\geq \frac{\Phi^*(\boldsymbol{x})}{\Phi(f_{\text{Unif}}(\boldsymbol{x}))}.
\end{align*}
Therefore, there exists $\tilde{\boldsymbol{x}}\in \{0,L\}^n$---namely, $\boldsymbol{x}''''$---with weakly higher  welfare approximation ratio than $\boldsymbol{x}$. Therefore,~(\ref{equation: funif approx}) holds. The theorem statement follows from the fact that the approximation ratio in Lemma \ref{lemma: UFS bound} is constructed by restricting agents to locations $\{0,L\}$.
\end{proof}

\section{Average Mechanism Results.}\label{Section: av mech}

\begin{proposition}\label{prop: average mech satisfies PF}
The average mechanism satisfies PF.
\end{proposition}

\begin{proof}
The average mechanism satisfies UFS and monotonicity. By Lemma~\ref{lem: ufs mono is PF}, it also satisfies PF.
\end{proof}

\begin{proposition}\label{prop: avg mech cost approx}
The average mechanism has a total cost approximation ratio of $2-\frac{2}{n}$.
\end{proposition}

\begin{proof} Let $f_{avg}$ denote the average mechanism. We show that for any location profile $\boldsymbol{x}\in [0,L]^n$ and some $k+1\geq \ceil{\frac{n+1}{2}}$, there exists some location profile $\tilde{\boldsymbol{x}}=(0,\dots,0,\tilde{x}_{k+1}\dots,\tilde{x}_n)$ such that
$$\frac{\Psi(f_{avg}(\tilde{\boldsymbol{x}}))}{\Psi^*(\tilde{\boldsymbol{x}})}\geq \frac{\Psi(f_{avg}(\boldsymbol{x}))}{\Psi^*(\boldsymbol{x})}.$$

Similar to the proof of Theorem~\ref{thm: UFS bound}, we order the agent labels such that $x_1\leq \dots \leq x_n$, let $i=med$ denote the median agent, and suppose without loss of generality that under $\boldsymbol{x}$, $x_{\ceil{\frac{n+1}{2}}}<f_{avg}(\boldsymbol{x})$.

First, consider the modified profile $\boldsymbol{x}'$ such that $x'_i=x_{med}$ for all $i\in S:=\{i:x_i<x_{med}\}$ and $x'_i=x_i$ for all $i\notin S$. In this profile, the median facility location does not change, and the average facility location moves to the right. The total cost corresponding to both facility locations decreases by $\sum_{i\in S}(x_{med}-x_i)$ from the agent movements. Also, as there are strictly more agents left of $f_{avg}(\boldsymbol{x})$ under $\boldsymbol{x}$, the total cost corresponding to the average mechanism increases from the facility moving to the right. We denote this change of total cost as $\Delta>0$. We therefore have
$$\frac{\Psi(f_{avg}(\boldsymbol{x}'))}{\Psi^*(\boldsymbol{x}')}=\frac{\Psi(f_{avg}(\boldsymbol{x}))-\sum_{i\in S}(x_{med}-x_i)+\Delta}{\Psi^*(\boldsymbol{x})-\sum_{i\in S}(x_{med}-x_i)}\geq \frac{\Psi( f_{avg}(\boldsymbol{x}))}{\Psi^*(\boldsymbol{x})}.$$

Now consider the modified profile $\boldsymbol{x}''$ where $x''_i=x'_{med}$ for all $i\in S:=\{i:x'_{med}<x'_i\leq f_{avg}(\boldsymbol{x}')\}$ and $x''_i=x_i'$ for all $i\notin S$. Again, the median facility location does not change, but now the average facility location moves to the left. Due to the agent movements, the total cost corresponding to the median mechanism decreases by $\sum_{i\in S}(x'_i-x'_{med})$, and the total cost corresponding to the average mechanism increases by $\sum_{i\in S}(x'_i-x'_{med})$. Also, the total cost corresponding to the average mechanism decreases by strictly less than $\sum_{i\in S}(x'_i-x'_{med})$ from the average facility location moving to the left by $\frac{1}{n}\sum_{i\in S}(x'_i-x'_{med})$. We therefore have
$$\frac{\Psi(f_{avg}(\boldsymbol{x}''))}{\Psi^*(\boldsymbol{x}'')}\geq\frac{\Psi(f_{avg}(\boldsymbol{x}'))}{\Psi^*(\boldsymbol{x}')-\sum_{i\in S}(x'_i-x'_{med})}\geq \frac{\Psi(f_{avg}(\boldsymbol{x}'))}{\Psi^*(\boldsymbol{x}')} \geq \frac{\Psi(f_{avg}(\boldsymbol{x}))}{\Psi^*(\boldsymbol{x})}.$$
Finally, we obtain $\tilde{\boldsymbol{x}}$ from $\boldsymbol{x}''$ by shifting all agents to the left by a distance of $x''_{med}$, leading to
$$\frac{\Psi(f_{avg}(\tilde{\boldsymbol{x}}))}{\Psi^*(\tilde{\boldsymbol{x}})}=\frac{\Psi(f_{avg}(\boldsymbol{x}''))}{\Psi^*(\boldsymbol{x}'')}\geq \frac{\Psi(f_{avg}(\boldsymbol{x}))}{\Psi^*(\boldsymbol{x})}.$$
Therefore, for some $k+1\geq \ceil{\frac{n+1}{2}}$, there exists $\tilde{\boldsymbol{x}}=(0,\dots,0,\tilde{x}_{k+1}\dots,\tilde{x}_n)$ with a weakly higher  total cost approximation ratio than $\boldsymbol{x}$. Now under $\tilde{\boldsymbol{x}}$, the median facility location is $0$, and thus the optimal total cost is $\sum^n_{i=k+1}\tilde{x}_i$. The average facility location is $\frac{1}{n}\sum^n_{i=k+1}\tilde{x}_i$, and thus the corresponding total cost is $\frac{k}{n}\sum^n_{i=k+1}\tilde{x}_i+\sum^n_{i=k+1}\tilde{x}_i-\frac{n-k}{n}\sum^n_{i=k+1}\tilde{x}_i=\frac{2k}{n}\sum^n_{i=k+1}\tilde{x}_i$. The total cost  approximation ratio of the theorem statement is obtained by setting $k=n-1$ and dividing the cost terms.
\end{proof}

\begin{proposition}
The average mechanism achieves the welfare  approximation ratio in Lemma~\ref{lemma: UFS bound}.
\end{proposition}
\begin{proof}
Let $f_{\text{avg}}$ denote the average mechanism. We prove that for any location profile $\boldsymbol{x}\in [0,L]^n$ there exists some profile $\tilde{\boldsymbol{x}}\in \{0,L\}^n$ such that
\begin{align}\label{equation: favg approx}
\frac{\Phi^*(\tilde{\boldsymbol{x}})}{\Phi(f_{\text{avg}}(\tilde{\boldsymbol{x}}))}  &\ge \frac{\Phi^*(\boldsymbol{x})}{\Phi(f_{\text{avg}}(\boldsymbol{x}))}.
\end{align}
This implies that
\[\max_{\boldsymbol{x}\in [0,L]^n}\frac{\Phi^*(\boldsymbol{x})}{\Phi(f_{\text{avg}}(\boldsymbol{x}))}=\max_{\boldsymbol{x}\in \{0,L\}^n}\frac{\Phi^*(\boldsymbol{x})}{\Phi(f_{\text{avg}}(\boldsymbol{x}))}.\]

Let the agent labels be ordered such that $x_1\le \ldots \le x_n$; let $i={med}$ denote the median agent.  Suppose without loss of generality that for odd $n$, we have $\boldsymbol{x} \ : \ x_{med}< f_{\text{avg}}(\boldsymbol{x})$ and for even $n$, we have $\boldsymbol{x} \ : \ x_{\frac{n}{2}+1}< f_{\text{avg}}(\boldsymbol{x})$. This is because~(\ref{equation: favg approx}) is trivially satisfied for odd $n$ if $x_{med}=f_{\text{avg}}(\boldsymbol{x})$, and it is satisfied for even $n$ if $x_{med}\leq f_{\text{avg}}(\boldsymbol{x})\leq x_{\frac{n}{2}+1}$.

First, consider the modified profile $\boldsymbol{x}'$ such that $x_i'=L$ for all $i\in S:=\{i:x_i\geq f_{\text{avg}}(\boldsymbol{x})\}$ and $x_i'=x_i$ for all $i\notin S$. In this modified profile, the {welfare-optimal} (median) location does not change, and the facility location under $f_{\text{avg}}$ moves towards the agents in $S$. Denoting this change in facility location as $\Delta>0$ and noting that $|S|<n-|S|$ due to the facility being located right of the {welfare-optimal} interval/median, the welfare under $f_{\text{avg}}$ decreases by $((n-|S|)-|S|)\Delta >0$ from the facility moving towards the $|S|$ agents at $L$ and away from the remaining $n-|S|$ agents. Due to the agent movements, the optimal welfare $\Phi^*(\boldsymbol{x}')$ and the welfare under $f$, $\Phi(f_{\text{avg}}(\boldsymbol{x}'))$ both decrease by the same amount ---namely, $\sum_{i\in S}(L-x_i)$--- relative to $\Phi^*(\boldsymbol{x})$ and $\Phi(f_{\text{avg}}(\boldsymbol{x}))$. We conclude that
\begin{align*}
\frac{\Phi^*(\boldsymbol{x}')}{\Phi(f_{\text{avg}}(\boldsymbol{x}'))}
&= \frac{\Phi^*(\boldsymbol{x})-\sum_{i\in S} (L-x_i)}{\Phi(f_{\text{avg}}(\boldsymbol{x}))-\sum_{i\in S} (L-x_i)-(n-2|S|)\Delta}\ge \frac{\Phi^*(\boldsymbol{x})}{\Phi(f_{\text{avg}}(\boldsymbol{x}))}.
\end{align*}
Now consider the modified profile $\boldsymbol{x}''$ such that $x_i''=0$ for all $i\in S':=\{i: x'_i<x_{med}\}$, for all $i\in S'':=\{i: x_{med}<x_i<f_{\text{avg}}(\boldsymbol{x}')\}$ and for $i=med$, and $x_i''=x_i'$ otherwise. The change in optimal welfare, which we will denote as $\Delta'_{opt}$, can be quantified by observing the agents' movements sequentially. The optimal welfare decreases by $\sum_{i\in S'}x_i$ from the agents of $S'$ moving to $0$. Next, the median agent (and {welfare-optimal} facility location) moving towards the $S'$ agents at $0$ causes the optimal welfare to decrease by $x_{med}\mathbb{I}_{n \text{ even}}$. Lastly, the remaining agents of $S''$ move towards the median at $0$, increasing the optimal welfare by $\sum_{i\in S''}x_i$. We therefore have
\begin{align}\label{equation: delta median}
\Delta'_{opt}=-\sum_{i\in S'}x_i-x_{med}\mathbb{I}_{n \text{ even}}+\sum_{i\in S''}x_i.
\end{align}
We next quantify the change in welfare corresponding to $f_{\text{avg}}$, which we denote as $\Delta'_{\text{avg}}$. The welfare decreases by $\sum_{i\in S'}x_i+x_{med}+\sum_{i\in S''}x_i$ from the agent movements, and increases by $(n-2|S|)\frac{1}{n}\sum_{i\in S'\cup \{med\}\cup S''}x_i$ from the facility moving towards the $n-|S|$ agents at $0$ and away from the $|S|$ agents at $L$. We therefore have
\begin{align}\label{equation: delta avg}
\Delta'_{\text{avg}}=-\sum_{i\in S'}x_i-x_{med}-\sum_{i\in S''}x_i+(n-2|S|)\frac{1}{n}\sum_{i\in S'\cup \{med\}\cup S''}x_i.
\end{align}
We now show that $\Delta'_{opt}>\Delta'_{avg}$ by subtracting Equations~(\ref{equation: delta median}) and~(\ref{equation: delta avg}). We first note that $|S''|=\frac{n}{2}-|S|$ for even $n$ and $|S''|=\frac{n-1}{2}-|S|$ for odd $n$. If $n$ is even, we have
\begin{align*}
\Delta'_{opt}-\Delta'_{avg}&=2\sum_{i\in S''}x_i-\frac{n-2|S|}{n}\left(\sum_{i\in S'\cup \{med\}}x_i+\sum_{i\in S''}x_i\right)\\
&\geq 2\sum_{i\in S''}x_i-\frac{2|S''|}{n}\left(\frac{n}{2}x_{med}+\sum_{i\in S''}x_i\right)\\
&=2\sum_{i\in S''}x_i-|S''|x_{med}-\frac{2|S''|}{n}\sum_{i\in S''}x_i\\
&=\left(\sum_{i\in S''}x_i-|S''|x_{med}\right)+\left(\sum_{i\in S''}x_i-\frac{2|S''|}{n}\sum_{i\in S''}x_i\right)\\
&\geq 0,
\end{align*}
where the first inequality is due to $x_{med}>x_i$ for all $i\in S'$, and we have $\sum_{i\in S''}x_i-|S''|x_{med}\geq 0$ due to $x_i>x_{med}$ for all $i\in S''$. Now if $n$ is odd, we have
\begin{align*}
\Delta'_{opt}-\Delta'_{avg}&=2\sum_{i\in S''}x_i+x_{med}-\frac{n-2|S|}{n}\left(\sum_{i\in S'}x_i+x_{med}+\sum_{i\in S''}x_i\right)\\
&\geq 2\sum_{i\in S''}x_i+x_{med}-\frac{2|S''|+1}{n}\left(\frac{n-1}{2}x_{med}+x_{med}+\sum_{i\in S''}x_i\right)\\
&=\left(\sum_{i\in S''}x_i-\frac{(2|S''|+1)(n-1)}{2n}x_{med}\right)+\left(x_{med}+\sum_{i\in S''}x_i\right)\left(1-\frac{2|S''|+1}{n}\right)\\
&=\left(\sum_{i\in S''}x_i-|S''|x_{med}-\frac{|S|}{n}x_{med}\right)+\left(x_{med}+\sum_{i\in S''}x_i\right)\left(\frac{2|S|}{n}\right)\\
&\geq 0.
\end{align*}
We have shown that $\Delta'_{opt}>\Delta'_{avg}$, meaning that we have 
\begin{align*}
\frac{\Phi^*(\boldsymbol{x}'')}{\Phi(f_{\text{avg}}(\boldsymbol{x}''))}
&= \frac{\Phi^*(\boldsymbol{x}')+\Delta'_{opt}}{\Phi(f_{\text{avg}}(\boldsymbol{x}'))+\Delta'_{avg}}\ge \frac{\Phi^*(\boldsymbol{x}') }{\Phi(f_{\text{avg}}(\boldsymbol{x}')) } \ge \frac{\Phi^*(\boldsymbol{x}) }{\Phi(f_{\text{avg}}(\boldsymbol{x})) }. 
\end{align*}
Therefore, there exists $\tilde{\boldsymbol{x}}\in \{0,L\}^n$---namely, $\boldsymbol{x}''$---with weakly higher  welfare approximation ratio than $\boldsymbol{x}$. Therefore,~(\ref{equation: favg approx}) holds. The proposition statement follows from the fact that the approximation ratio in Lemma \ref{lemma: UFS bound} is constructed by restricting agents to locations $\{0,L\}$.
\end{proof}

\section{Other Total Cost Approximation Results}\label{appsec: totalcostapprox}
\begin{proposition}
The Nash mechanism has a total cost approximation ratio of at least $2-\frac{2}{n}$.
\end{proposition}
\begin{proof}
As in the proof of Proposition~\ref{prop: avg mech cost approx}, we obtain the lower bound of $2-\frac{2}{n}$ from the location profile where $n-1$ agents at located at $0$ and $L$ agent is located at $L$. As the Nash mechanism satisfies UFS \cite{LAW21}, it places the facility at $\frac{L}{n}$, leading to the total cost approximation of $2-\frac{2}{n}$. 
\end{proof}
\begin{proposition}
The Nash mechanism has a total cost approximation ratio of at most $\frac{n}{2}$.
\end{proposition}
\begin{proof}
Suppose without loss of generality that the leftmost agent is located at $0$ and the rightmost agent is located at $c$. We first show that the Nash mechanism guarantees that the total cost is at most $\frac{nc}{2}$. The proof is a modification of the Proof of Lemma~7 in \cite{LAW21} for the total cost objective. To prove this, we show that for any such location profile $\boldsymbol{x}$, there exists some profile $\tilde{\boldsymbol{x}}:=(\underbrace{0,\dots,0}_{\floor{\frac{n}{2}}},\underbrace{c,\dots,c}_{\ceil{\frac{n}{2}}})$ such that $\Psi(f_{Nash}(\tilde{\boldsymbol{x}}))\geq \Psi(f_{Nash}(\boldsymbol{x}))$. Let $n-k:=|\{i:x_i > f_{Nash}(\boldsymbol{x})\}|$ and $k:=|\{i:x_i \leq f_{Nash}(\boldsymbol{x})\}|$. For this proof, we also suppose that the agents are ordered such that $x_1\leq \dots \leq x_n$, and without loss of generality that $n-k\geq k$.

First, consider the modified profile $\boldsymbol{x}'$ such that $x'_i=0$ for all $i\in S:=\{i: x_i \leq f_{Nash}(\boldsymbol{x})$ and $x'_i=x_i$ for all $i\notin S$. Here, we have moved all agents initially located on or to the left of the facility to $0$. Let $\Delta$ be the change in the facility location as a result of this transformation. As the Nash mechanism is weakly monotonic, we have $\Delta\leq 0$. The net change in total cost is $\sum_{i\in S}x_i -\Delta ((n-k)-k)\geq 0$.

We now have $\boldsymbol{x}'=(\underbrace{0,\dots,0}_k,x'_{k+1},\dots,x'_n)$. If $k=\floor{\frac{n}{2}}$, the following transformation can be skipped. Otherwise, we suppose that $k<\floor{\frac{n}{2}}$. As shown in the proof of Lemma~7 in \cite{LAW21}, we have $f_{Nash}(\boldsymbol{x}')\geq \frac{x'_{k+1}}{2}$. We shift $x'_{k+1}$ to $0$, and denote $\Delta'\leq 0$ as the change in facility location. The net change in total cost is $[(f_{Nash}(\boldsymbol{x}')-0)-(x'_{k+1}-f_{Nash}(\boldsymbol{x}'))]-\Delta'[(n-k-1)-(k+1)]\geq 0$. The first term is non-negative as $f_{Nash}(\boldsymbol{x}')\geq \frac{x'_{k+1}}{2}$, and the second term is non-negative as $k+1\leq \floor{\frac{n}{2}}$. To form location profile $\boldsymbol{x}''$, we iteratively shift agent locations $x'_{k+1},\dots,x'_{\floor{\frac{n}{2}}}$ to $0$, and the same arguments can be applied to show that the total cost does not decrease.

Finally, we transform $\boldsymbol{x}''=(\underbrace{0,\dots,0}_{\lfloor \frac{n}{2}\rfloor},x''_{\lceil \frac{n}{2}\rceil},\dots,x''_n)$ to $\tilde{\boldsymbol{x}}:=(\underbrace{0,\dots,0}_{\floor{\frac{n}{2}}},\underbrace{c,\dots,c}_{\ceil{\frac{n}{2}}})$ by shifting the agents at $x''_{\lceil \frac{n}{2}\rceil},\dots,x''_n$ to $c$. Let $\Delta''\geq 0$ be the change in facility location. By Lemma~4 of \cite{LAW21}, we have $\Delta'' \leq \max_{i\in \{\lceil \frac{n}{2}\rceil,\dots,n\}}|c-x''_i|$, and hence the change in total cost is $\sum_{i=\lceil \frac{n}{2}\rceil}^n(c-x''_i)-\Delta'' (\lceil \frac{n}{2}\rceil - \lfloor \frac{n}{2}\rfloor)\geq 0$.

By Lemma~5 of \cite{LAW21}, $f_{Nash}(\tilde{\boldsymbol{x}})=\frac{c}{2}$ if $n$ is even and $f_{Nash}(\tilde{\boldsymbol{x}})=\frac{c}{2}-\frac{c}{2n}+\frac{1}{n}>\frac{c}{2}$ if $n$ is odd. Therefore the Nash mechanism guarantees at most $\frac{nc}{2}$ total cost when the leftmost agent is located at $0$ and the rightmost agent is located at $c$. Now under such a profile, the optimal total cost is at least $c$, from placing the facility with $n-1$ agents at $0$, and the remaining agent at $c$. Dividing these terms gives us an upper bound of $\frac{n}{2}$ for the Nash mechanism's total cost approximation ratio.
\end{proof}

\begin{proposition}
The midpoint mechanism has a total cost approximation ratio of $\frac{n}{2}$.
\end{proposition}
\begin{proof}
Suppose without loss of generality that the leftmost agent is located at $0$, and the rightmost agent is located at $c$. Clearly, the optimal total cost must be at least $c$. An upper bound for the total cost corresponding to the midpoint mechanism is $\frac{nc}{2}$, achieved when the agents are only located  at $0$ or $c$. Dividing these terms gives us an upper bound on the total cost approximation ratio of $\frac{n}{2}$. This is matched by the location profile where $n-1$ agents at located at $0$ and $1$ agent is located at $c$.
\end{proof}
%
%
%

\section*{Acknowledgments.}
The authors thank Rupert Freeman, Arindam Pal, Pedro Pablo P{\'e}rez Velasco, Herv{\'e} Moulin, and Mashbat Suzuki for valuable comments.  



\end{document}